\documentclass[journal,12pt,onecolumn,draftclsnofoot]{IEEEtran}

\DeclareMathAlphabet\mathbfcal{OMS}{cmsy}{b}{n}
\usepackage{hyperref} 
\usepackage[english]{babel}
\usepackage{amsmath}
\usepackage{graphicx}
\usepackage{dsfont}
\usepackage{amsthm}
\usepackage{amssymb}
\usepackage{mdframed}
\usepackage{subfigure}
\usepackage[usenames,dvipsnames]{xcolor}
\usepackage{soul}
\usepackage{enumerate}
\usepackage{xspace}
\usepackage[disable,colorinlistoftodos]{todonotes}

\newcommand{\idle}{\mathrm{i}}
\newcommand{\retx}{\mathrm{x}}
\newcommand{\new}{\mathrm{n}}

\usepackage[linesnumbered,ruled]{algorithm2e}

\usepackage{lscape}
\usepackage{bbm}
\usepackage{cite}
\DeclareMathOperator*{\argmin}{arg\,min} 
\newtheorem{theorem}{Theorem}
\newtheorem{problem}{Problem}

\newtheorem{lemma}{Lemma}

\newtheorem{prop}{Proposition}

\newcommand{\Exp}[1]{\mathbb{E}\left[ #1 \right]} 
\newcommand{\Pt}{\mathrm{P}}

\title{Average Age of Information with Hybrid ARQ under a Resource Constraint}

\author{
\IEEEauthorblockN{Elif Tu\u{g}\c{c}e Ceran, Deniz G{\"u}nd{\"u}z, and Andr\'as Gy\"orgy}\\
\IEEEauthorblockA{Department of Electrical and Electronic Engineering \\Imperial College London\\
Email: \{e.ceran14, d.gunduz, a.gyorgy\}@imperial.ac.uk
}}

\begin{document}
\maketitle

\begin{abstract}
Scheduling  the transmission of status updates over an error-prone communication channel is studied in order to minimize the long-term average \textit{age of information} at the destination under a constraint on the average number of transmissions at the source node. After each transmission, the source receives an instantaneous ACK/NACK feedback, and decides on the next update without prior knowledge on the success of future transmissions. The optimal scheduling policy is first studied under different feedback mechanisms when the channel statistics are known; in particular, the standard automatic repeat request (ARQ) and hybrid ARQ (HARQ) protocols are considered. Structural results are derived for the optimal policy under HARQ, while the optimal policy is determined analytically for ARQ. For the case of unknown environments, an average-cost reinforcement learning algorithm is proposed that learns the system parameters and the transmission policy in real time. The effectiveness of the proposed methods is verified through numerical results. \makeatletter{\renewcommand*{\@makefnmark}{}
 \footnotetext{Part of this work was presented at the  IEEE Wireless Communications and Networking Conference, Barcelona, Spain, April 2018 \cite{wcnc_paper}.}\makeatother}
\end{abstract}


\section{Introduction} 

Motivated by the growing interest in timely delivery of  information in status update systems, the \emph{age of information (AoI)} has been introduced as a performance measure to quantify data staleness at the receiver \cite{Altman2010, Kaul2011, Kaul2012}. Consider a source node that samples an underlying time-varying process and sends the sampled status of the process over an imperfect communication channel that introduces delays. The AoI characterizes the data staleness (or tardiness) at the destination node, and it is defined as the time that has elapsed since the most recent status update available at the destination was generated. Different from classical performance measures, such as the delay or throughput, AoI jointly captures the latency in transmitting updates and the rate at which they are delivered.

Our goal in this paper is to minimize the average AoI at the destination taking into account \emph{retransmissions} due to errors over the noisy communication channel. Retransmissions
are essential for providing reliability of status updates over error-prone channels, particularly in wireless settings.  Here, we analyze the AoI for both the standard ARQ and hybrid ARQ (HARQ) protocols.

In the HARQ protocol, the receiver combines information from all previous transmission attempts of the same packet in order to increase the success probability of decoding \cite{hybrid2001}, \cite{harq2003}, \cite{Lagrange2010}.  The exact relationship between the probability of error and the number of retransmission attempts varies depending on the channel conditions and the particular HARQ method employed \cite{hybrid2001}, \cite{harq2003}, \cite{Lagrange2010}. In general, the probability of successful decoding increases with each transmission, but the AoI of the received packet also increases. Therefore, there is an inherent trade-off between retransmitting  previously failed status information with a lower error probability, or sending a fresh status update with higher error probability. We address this trade-off between the success probability and the freshness of the status update to be transmitted, and develop scheduling policies to minimize the expected average AoI. 

In the standard ARQ protocol, if a packet cannot be decoded, it is retransmitted until successful reception. Note, however, that, when optimizing for the AoI, there is no point in retransmitting the same packet, since a newer packet with more up-to-date information is available at the sender at the time of retransmission. Thus, after the reception of a NACK feedback, the actual packet is discarded, and the most recent status of the underlying process is transmitted (the exact timing of the transmission may depend on the feedback, i.e., on the success history of previous transmissions). A scheduling  policy to decide whether to stay idle or transmit a status update  should be designed considering a resource constraint on the average number of transmissions.



We develop scheduling policies for both the HARQ and the standard ARQ protocols to minimize the expected average AoI under a constraint on the average number of transmissions, which is motivated by the fact that sensors sending status updates have usually limited energy supplies (e.g., are powered via energy harvesting \cite{Gunduz2014}); and hence, they cannot afford to send an unlimited number of updates, or increase the signal-to-noise-ratio in the transmission. First, we assume that the success probability before each transmission attempt is known (which, in the case of HARQ, depends on the number of previous unsuccessful transmission attempts); and therefore, the source node can judiciously decide when to retransmit and when to discard a failed packet and send a fresh update. Then, we consider transmitting status updates over an unknown channel, in which case the success probabilities of transmission attempts are not known \emph{a priori}, and must be learned in an online fashion. This latter scenario can model sensors embedded in unknown or time-varying environments. We employ reinforcement learning (RL) algorithms to balance exploitation and exploration in an unknown environment, so that the source node can quickly learn the environment based on the ACK/NACK feedback signals, and can adapt its scheduling policy accordingly, exploiting its limited resources in an efficient manner.


The main contributions of this paper are as follows:

\begin{itemize}
\item Average AoI is studied under a long-term average resource constraint imposed on the transmitter, which limits the average number of transmissions. 
\item Both retransmissions and pre-emption following a failed transmission are considered, corresponding, respectively, to the HARQ and ARQ protocols, and the structure of the optimal policy is determined in general.
\item The optimal preemptive transmission policy for the standard ARQ protocol is shown to be a threshold-type randomized policy, and is derived in closed-form.
\item An average-cost RL algorithm; in particular, \textit{average-cost SARSA with softmax},  is proposed to learn the optimal scheduling decisions when the transmission success probabilities are unknown. 
\item Extensive numerical simulations are conducted in order to show the effect of feedback, resource constraint and ARQ or HARQ mechanisms on the freshness of the data. 
\end{itemize}

\subsection{Related Work}
\label{sec:related}

Most of the earlier work on AoI consider queue-based models, in which the status updates arrive at the source node randomly following a memoryless Poisson process, and are stored in a buffer before being transmitted to the destination over a noiseless channel \cite{Kaul2011, Kaul2012}. Instead, in the so-called \emph{generate-at-will} model, \cite{Altman2010,Tan2015,Kadota2016,hsuage2017,Sun2017_tran}, also adopted in this paper,  the status  of the underlying process can be sampled at any time by the source node. 


A constant packet failure probability for a status update system is investigated for the first time in \cite{Chen2016}, where status updates arrive according to a Poisson process, while the transmission time for each packet is exponentially distributed. Fast-come-first-served (FCFS) scheduling is analyzed and it is shown that packet loss and large queuing delay due to old packets in the queue result in an increase in the AoI. Different scheduling decisions at the source node are investigated; including the last-come-first-served (LCFS)  principle, which always transmits the most up-to-date packet, and retransmissions with preemptive priority, which preempts the current packet in service when a new packet arrives. 

Broadcasting of status updates to multiple receivers over an unreliable broadcast channel is considered in \cite{Kadota2016}. A low complexity sub-optimal scheduling policy is proposed when the AoI at each receiver and the transmission error probabilities to all the receivers are known. However,  only  work-conserving policies are considered in \cite{Kadota2016}, which update the information at every time slot, since no constraint is imposed on the number of updates. Optimizing the scheduling decisions with multiple receivers over a perfect channel is investigated in \cite{hsuage2017}, and it is shown that there is an optimal scheduling algorithm that is of threshold-type. To our knowledge, the latter is the only prior work in the literature which applies RL in the AoI framework. However, their goal is to learn the data arrival statistics, and it does not consider either an unreliable communication link or HARQ. Moreover, we employ an average-cost RL method, which has significant advantages over discounted-cost methods, such as \emph{Q-learning} \cite{Mahadevan1996}.

The AoI in the presence of HARQ has been considered in  \cite{Parag2017, Najm2017} and  \cite{Yates2017}. In \cite{Parag2017} the affect of design decisions, such as the length of the transmitted codewords, on the average AoI is analyzed. The status update system is modeled as an M/G/1/1 queue in \cite{Najm2017}; however, no resource constraint is considered, and the status update arrivals are assumed to be memoryless and random, in contrast to our work, which considers the \emph{generate-at-will} model. Moreover, a specific coding scheme is assumed in \cite{Najm2017}, namely MDS (maximum distance separable) coding, which results in a particular formula for the successful decoding probabilities, whereas we allow general functions for these probabilities. From a queuing-systems perspective, our model can be considered as a G/G/1/1 queue with optimization of packet arrivals and pre-emption. In \cite{Yates2017}, HARQ is considered in a zero-wait system, where as soon as an update is successfully transmitted to the destination, the source starts transmitting a new status update, as no resource constraint or pre-emption is taken into account.

In \cite{Altman2010} and \cite{fenni2012}, the receiver can choose to update its status information by downloading an update over one of the two available channels, a free yet unreliable channel, modeling a Wi-Fi connection, and a reliable channel with a cost, modeling a cellular connection. 
Although the Lagrangian formulation of our constrained optimization problem for the standard ARQ protocol is similar to the one considered in \cite{Altman2010}, our problem is more complicated due to several reasons: they have not considered the effect of retransmissions or any algorithm that learns the unknown system parameters, and even without these complications, we need to determine the Lagrange multiplier corresponding to the given constraints, while it is given in \cite{Altman2010}.
 
To the best of our knowledge, this is the first work in the literature that addresses a status update system with HARQ in the presence of resource constraints. In addition, no previous work has studied the average AoI over a channel with unknown error probabilities, and employed an average-cost RL algorithm.  


\section{System Model and Problem Formulation}
\label{sec:system}

\begin{figure}
\centering
\includegraphics[scale=0.5]{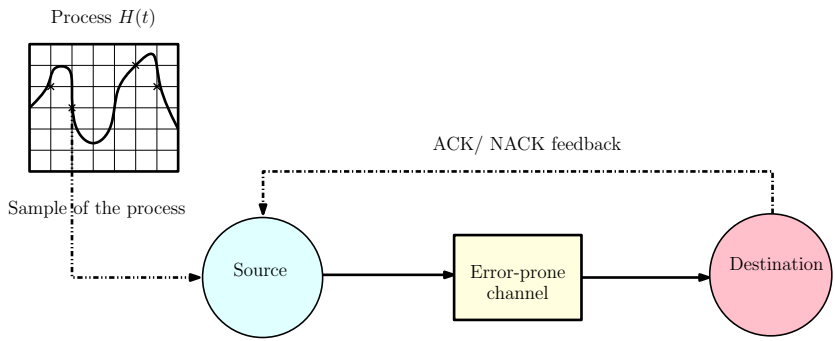}
\caption{System model of a status update system over an error-prone point-to-point link in the presence of ACK/NACK feedback from the destination.}
\label{fig:system}
\end{figure}

We consider a time-slotted status update system over an error-prone 
communication link (see Figure~\ref{fig:system}). The source monitors an underlying time-varying process, and can generate a status update at any time slot; known as the \emph{generate-at-will} model \cite{Sun2017_tran}. The status updates are communicated from the source node to the destination over a time-varying channel. Each transmission attempt of a status update takes constant time, which is assumed to be equal to the duration of one time slot. 	Throughout the paper, we will normalize all time durations by the duration of one time slot.

We assume that the 
channel changes randomly from one time slot to the next in an independent and identically distributed fashion, and the channel state information is available only at the destination node. We further assume the availability of an error- and delay-free single-bit feedback from the destination to the source node for each transmission attempt. Successful receipt of a status update is acknowledged by an ACK signal, while a NACK signal is sent in case of a failure. In the classical ARQ protocol, a packet is retransmitted after each NACK feedback, until it is successfully decoded (or a maximum number of allowed retransmissions is reached), and the received signal is discarded after each failed transmission attempt. Therefore, the probability of error is the same for all retransmissions. However, in the AoI framework there is no point in retransmitting a failed out-of-date status packet if it has the same error probability as that of a fresh update. Hence, we assume that if the ARQ protocol is adopted,  the source always removes  failed packets and transmits a fresh status update.
On the other hand, if the HARQ protocol is used, the received signals from all previous transmission attempts for the same packet are combined for decoding. Therefore, the probability of error decreases with every retransmission. In general, the error probability of each retransmission attempt depends on the particular combination technique used by the decoder, as well as on the channel conditions \cite{hybrid2001}. 

\begin{figure}
\centering
\includegraphics[scale=0.55]{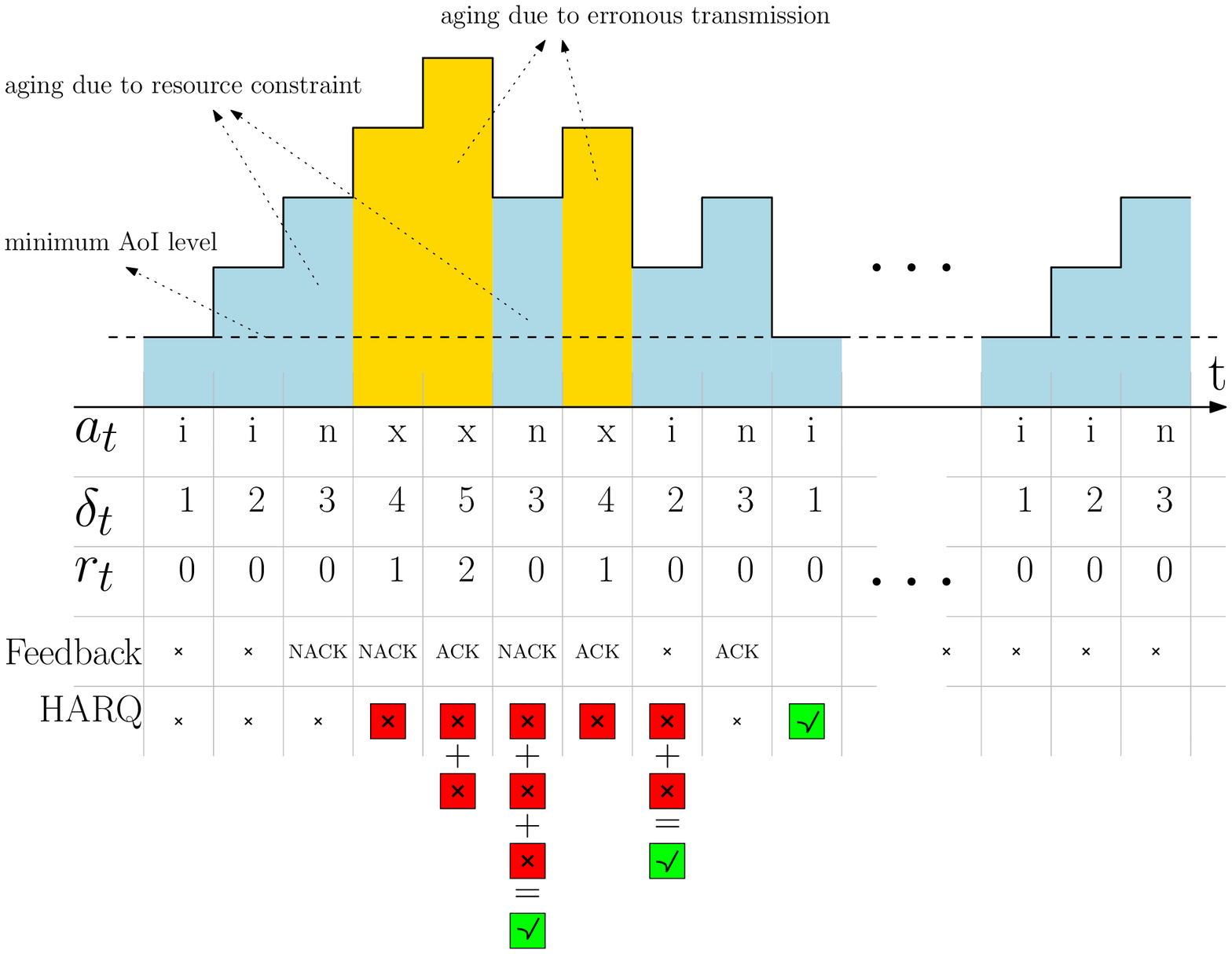}
\caption{Illustration of the AoI in a slotted status update system with HARQ. ($\delta_t,r_t$) represents the state of the system and the action is chosen based on the state ($\delta_t,r_t$) and denoted by $a_t$. Packets with decoding errors (represented by red squares) are stored in the receiver and combined to decode the information successfully (represented by green squares). }
\label{fig:slotted}
\end{figure}


AoI measures the timeliness of the information at the receiver. It is defined as the number of time slots elapsed since the generation of the most up-to-date packet successfully decoded at the receiver. Formally, denoting the latter generation time  for any time slot $t$ by $U(t)$, the AoI, denoted by $\delta_t$, is defined as
\begin{align}
\delta_t\triangleq t-U(t).
\end{align}
We assume that a transmission decision is made at the beginning of each slot. The AoI increases by one when the transmission fails, while it decreases to one in the case of ARQ, or to the number of retransmissions plus one in the case of HARQ, when a status update is successfully decoded
(the minimum age is set to $1$ to reflect that the transmission is one slot long).

The probability of error after $r$ retransmissions, denoted by $g(r)$, depends on $r$ and the particular HARQ scheme used for combining multiple transmission attempts (an empirical method to estimate $g(r)$ is presented in \cite{harq2003}). As in any reasonable HARQ strategy, we assume that $g(r)$ is non-increasing in the number of retransmissions $r$; that is, $g(r_1) \geq g(r_2)$ for all $r_1 \leq r_2$. For simplicity, we assume that $0<g(0)<1$, that is, the channel is noisy and there is a possibility that the first transmission is successful (if $g(0)=0$, the problem becomes trivial, while $g(0)<1$ can be easily relaxed to the condition that there exists an $r$ such that $g(r)<1$). Also, we will denote the maximum number of retransmissions by $r_{max}$, which may take the value $\infty$, unless otherwise stated. However, if $g(r)=0$ for some $r$ (i.e., a packet is always correctly decoded after $r$ retransmissions), we set $r_{max}$ to be the smallest such $r$. Finally note that standard HARQ methods only allow a finite maximum number of retransmissions \cite{IEEEstandard}.

For any time slot $t$, let $\delta_t\in \mathds{Z}^+$ denote the AoI at the beginning of the time slot and $r_t \in \{0,\ldots,r_{max}\}$ denote the number of previous transmission attempts of the same packet. Then the state of the system can be described by  $s_t \triangleq (\delta_t,r_t)$. At each time slot, the source node takes one of the three actions, denoted by $a \in \mathcal{A}$, where $\mathcal{A}=\{\idle,\new,\retx\}$:  (i) remain idle ($a=\idle$); (ii) transmit a new status update ($a=\new$); or (iii) retransmit the previously failed update ($a=\retx$).  The  evolution of AoI for a slotted status update system is illustrated in Figure~\ref{fig:slotted}.

Note that if  no resource constraint is imposed on the source, remaining idle is clearly suboptimal since it does not contribute to decreasing the AoI. However, continuous transmission is typically not possible in practice due to energy or interference constraints. Accordingly, we impose a constraint on the average number of transmissions, and we require that the long-term average number of transmission do not exceed $C_{max} \in (0,1]$ (note that $C_{max}=1$ corresponds to the case in which transmission is allowed in every slot).

This leads to the \emph{constrained Markov decision process} (CMDP) formulation, defined by the 5-tuple $\big(\mathcal{S}, \mathcal{A}, \Pt, c, d\big)$ \cite{Altman}: The countable set of states  $(\delta ,r) \in \mathcal{S} $ and the finite action set $\mathcal{A}=\{\idle,\new,\retx\}$ have already been defined. $\Pt$ refers to the transition function, where  $\Pt(s'|s,a) = \Pr(s_{t+1}=s' \mid s_t = s, a_t=a)$ is the probability that action ${\displaystyle a}$  in state ${\displaystyle s}$  at time ${\displaystyle t}$  will lead to state ${\displaystyle s'}$ at time ${\displaystyle t+1}$, which will be explicitly defined in (\ref{eq:transitions}). The cost function $c: \mathcal{S} \times \mathcal{A} \rightarrow \mathbbm{R}$, is the AoI at the destination, and is defined as $c((\delta,r),a)=\delta$ for any $(\delta,r)\in \mathcal{S}$, $a \in \mathcal{A}$, independently of action $a$. The transmission cost $d:\mathcal{S} \times \mathcal{A} \rightarrow \mathbbm{R}$ is independent of the state and depends only on the action $a$, where $d = 0$ if $a=\idle$, and $d=1$ otherwise. Since our goal is to minimize the AoI subject to a constraint on the average transmission cost, the corresponding problem is a CMDP. 

A policy is a sequence of decision rules $\pi_t: (\mathcal{S} \times \mathcal{A})^t \to [0,1]$, which maps the past states and actions and the current state to a distribution over the actions; that is, after the state-action sequence $s_1,a_1,\ldots,s_{t-1},a_{t-1}$,  action $a$ is selected in state $s_t$ with probability $\pi_t(a_t|s_{1},a_{1},\ldots,$ $s_{t-1},a_{t-1},s_t)$. We  use $s_t^{\pi}=(\delta_t^{\pi},r_t^{\pi})$ and $a_t^{\pi}$ to denote the sequences of states and actions, respectively, induced by policy $\pi=\{\pi_t\}$. 
A policy $\pi=\{\pi_t\}$ is called \emph{stationary} if the distribution of the next action is independent of the past states and actions given the current state, and time invariant; that is, with a slight abuse of notation, $\pi_t(a_t|s_{1},a_{1},\ldots,s_{t-1},a_{t-1},s_t)=\pi(a_t|s_t)$ for all $t$ and $(s_i,a_i) \in \mathcal{S} \times \mathcal{A}$, $i=1, \ldots, t$. Finally, a policy is said to be deterministic  if it chooses an action with probability one; with a slight abuse of notation, we use $\pi(s)$ to denote the action taken with probability one in state $s$ by a stationary deterministic policy.

Let $J^{\pi}(s_0)$ and $C^{\pi}(s_0)$ denote the infinite horizon average age and the average number of transmissions, respectively, when policy $\pi$ is employed with initial state $s_0$. Then the CMDP optimization problem can be stated as follows:
\begin{problem}
\begin{align}
& \mathrm{ Minimize }  ~J^{\pi}(s_0) \triangleq \limsup_{T\rightarrow \infty }\frac{1}{T}\ \Exp{\sum_{t=1}^T{\delta^{\pi}_t}\Big|s_0}, \label{eq:cost}\\
& \mathrm{subject}~\mathrm{to }  ~C^{\pi}(s_0) \triangleq \limsup_{T\rightarrow \infty }\frac{1}{T}\ \Exp{\sum_{t=1}^T{\mathbbm{1}[a^{\pi}_t \neq \idle]  }\Big|s_0}\leq C_{max}. \label{eq:constraint}
\end{align} 
\label{problem}
\end{problem} 
A policy $\pi$ that is a solution of the above minimization problem is called optimal, and we are interested in finding optimal policies. 
Without loss of generality, we assume that the sender and the receiver are synchronized at the beginning of the problem, that is, $s_0=(1,0)$; and we omit $s_0$ from the notation for simplicity.

Before formally defining the transition function $\Pt$ in our AoI problem, we present a simple observation that simplifies $\Pt$: Retransmitting a packet immediately after a failed attempt is better than retransmitting it after waiting for some slots. This is true since  waiting increases the age, without increasing the success probability. The difference in the waiting time is illustrated in Figure~\ref{fig:proof} for a simple scenario, where  the first transmission of a status update results in a failure, while the retransmission is successful. 

\begin{prop}
\label{p1}
For any policy $\pi$ there exists another policy $\pi'$ (not necessarily distinct from $\pi$) such that $J^{\pi'}(s_0) \le J^{\pi}(s_0)$, $C^{\pi'}(s_0) \le C^{\pi}(s_0)$, and $\pi'$ takes a retransmission action only following a failed transmission, that is, the probability  $Pr(a^{\pi'}_{t+1}=\retx|a^{\pi'}_t=\idle)= 0$.
\end{prop}

\begin{figure}
\centering
\includegraphics[scale=0.42]{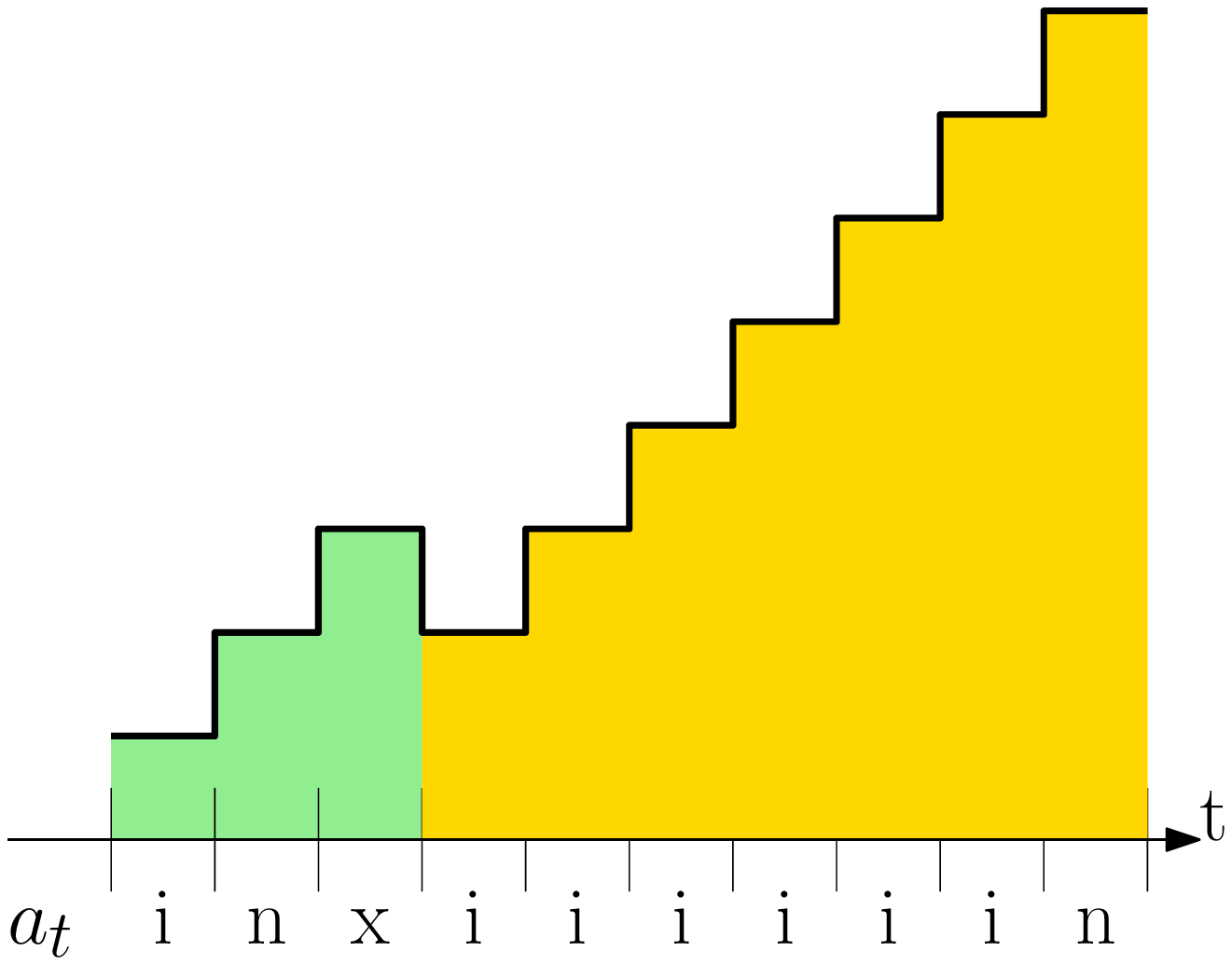}
\includegraphics[scale=0.42]{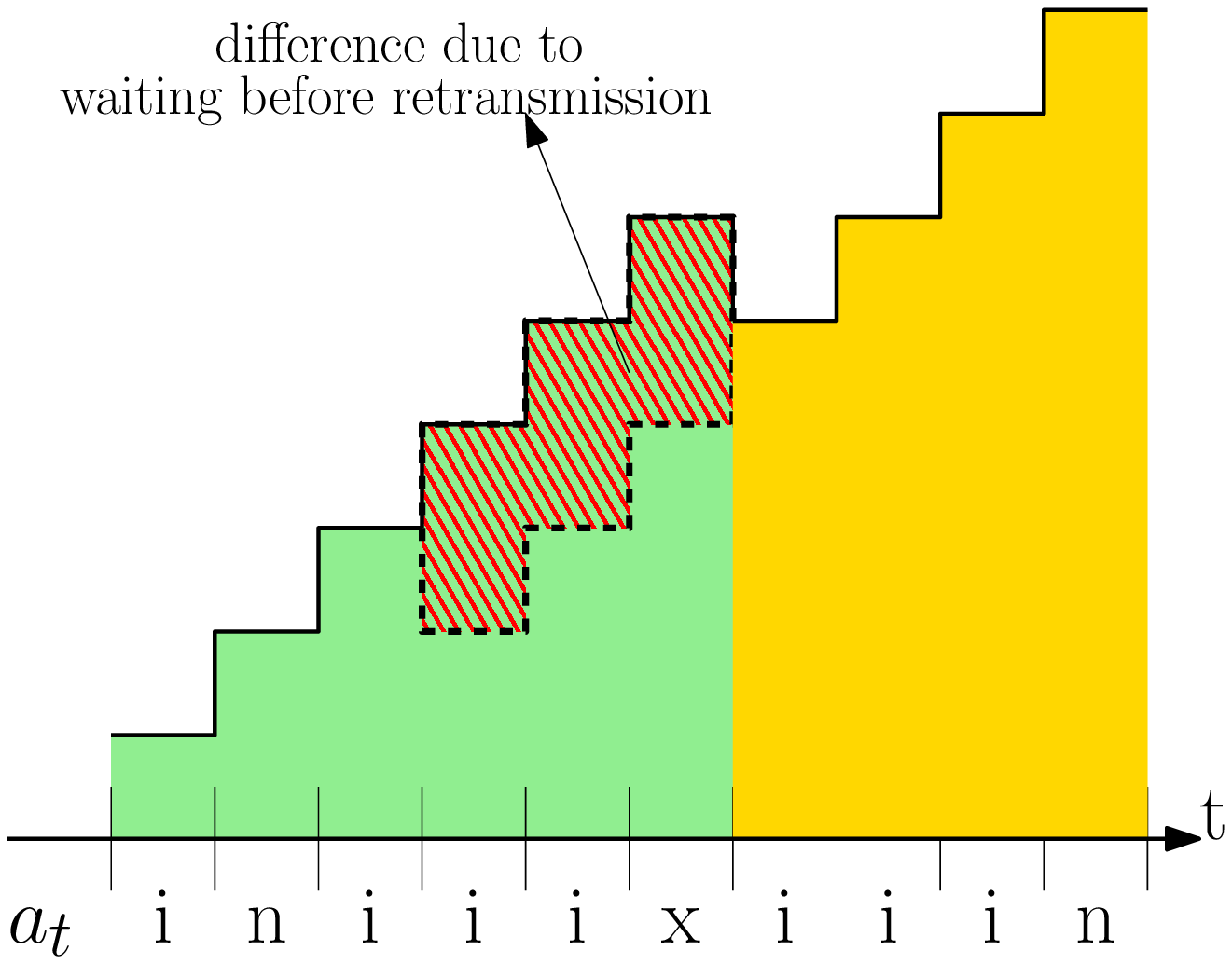}
\caption{The difference of the AoI for policies without and with idle slots before retransmissions. The figure on the left shows the evolution of age (height of the bars) when retransmission occurs immediately after an error in transmission whereas the figure on the right represents the evolution of age when retransmission occurs after some idle slots.}
\label{fig:proof}
\end{figure}

The  transition probabilities are given as follows (omitting the parenthesis from the state variables $(\delta,r)$):
\begin{equation}\label{eq:transitions}
\begin{split}
&\Pt(\delta+1, 0 | \delta, r, \idle) = 1, \\
&\Pt(\delta +1, 1 | \delta, r, \new) = g(0),  \\
&\Pt(1, 0 | \delta, r, \new) = 1 - g(0),  \\
&\Pt(\delta+1, r+1 | \delta, r,\retx) = g(r),\\
&\Pt(r+1,0 | \delta,r,\retx) = 1 - g(r),
\end{split}
\end{equation}
and $\Pt(\delta',r'|\delta,r,a)=0$  otherwise.  Note that the above equations set the retransmission count to $0$ after each successful transmission, and it is not allowed to take a retransmission action in states where the transmission count is $0$. Also, the property in Proposition~\ref{p1} is enforced by the first equation in \eqref{eq:transitions}, that is, $\Pt(\delta+1, 0| \delta, r, \idle) = 1$ (since retransmissions are not allowed in states $(\delta,0)$).
Since the starting state is $(1,0)$, it also follows that the state set of our CMDP can be described as
\begin{equation}
\label{eq:S}
\mathcal{S}=\{(\delta,r): r < \min\{\delta, r_{max}+1\}, \delta, r \in \mathbb{N} \}~.
\end{equation}


\section{Lagrangian Relaxation and the Structure of the Optimal Policy}
\label{sec:structure}

In this section, we derive the structure of the optimal policy for Problem~\ref{problem} based on \cite{sennott_1993}. A detailed treatment of finite state-finite action CMDPs is considered in \cite{Altman}, but here we need more general results that apply to countable state spaces. These results require certain technical conditions; 
roughly speaking, there must exist a deterministic policy that satisfies the transmission constraint while maintaining a finite average AoI, and any ``reasonable'' policy must induce a positive recurrent Markov chain. The precise formulation of the requirements is given in Appendix~\ref{Append_assumptions}, wherein Proposition~\ref{lemma_assump} shows that the conditions of \cite{sennott_1993} are satisfied for Problem~\ref{problem}. Given this result, we follow \cite{sennott_1993}  to characterize the optimal policy.

While there exists a stationary and deterministic optimal policy for countable-state finite-action average-cost MDPs \cite{sennott_1989,Puterman_book,Bertsekas2000}, this is not necessarily true for CMDPs \cite{sennott_1993,Altman}.  To solve the CMDP, we start with rewriting the problem in its Lagrangian form. The average Lagrangian cost of a policy $\pi$ with Lagrange multiplier $\eta \ge 0$ is defined as
\begin{align}
L^{\pi}_{\eta}=\lim_{T\rightarrow \infty }\frac{1}{T}\left(\Exp{\sum_{t=1}^T{\delta^{\pi}_t}}+\eta \Exp{\sum_{t=1}^T{\mathbbm{1}[a^{\pi}_t\neq \idle]}}\right),
\label{eq:lag}
\end{align}
and, for any $\eta$, the optimal achievable cost $L^*_{\eta}$ is defined as $L^*_{\eta}\triangleq \min_{\pi} L^{\pi}_{\eta}$. If the constraint on the transmission cost is less than one (i.e., $C_{max}<1$), then we have $\eta>0$, which will be assumed throughout the paper.\footnote{If $C_{max}=1$, a transmission (new update or retransmission) is allowed in every time slot, and instead of a CMDP we have an infinite state-space MDP with unbounded cost. Then it follows directly from part (ii) of the Theorem of \cite{sennott_1989} (whose conditions can be easily verified for our problem) that there exists an optimal stationary policy that satisfies the Bellman equations. In this paper we concentrate on the more interesting constrained case, while the derivation of this result is relegated to Appendix~\ref{sec:unconstrained}.} This formulation is equivalent to an unconstrained countable-state average-cost MDP in which the instantaneous overall cost is  $\delta_t+\eta\mathbbm{1}[a^{\pi}_t \neq \idle]$. A policy $\pi$ is called $\eta$-optimal if it achieves $L^*_\eta$. Since the assumptions of Proposition~3.2 of \cite{sennott_1993} are satisfied by Proposition~\ref{lemma_assump} in Appendix~\ref{Append_assumptions}, the former implies  that there exists a function $h_{\eta}(\delta,r)$, called the \textit{differential cost function}, satisfying
\begin{align}
\label{eq:Bellman}
h_{\eta}(\delta,r)+L^*_\eta&=\min_{a\in\{\idle,\new,\retx\}}\big(\delta+\eta \cdot \mathbbm{1}[a \neq \idle]+\Exp{h_{\eta}(\delta',r')}\big), 
\end{align}
called the \emph{Bellman optimality equations}, for all states $(\delta,r) \in \mathcal{S}$, where $(\delta',r')$ is the next state obtained from $(\delta,r)$ after taking action $a$. Furthermore, Proposition 3.2 of \cite{sennott_1993} also implies that the function $h_\eta$ satisfying \eqref{eq:Bellman} is unique up to an additive factor, and with 
selecting this additive factor properly, the \textit{differential cost function} also satisfies
\[
h_\eta(\delta,r) = \Exp{\sum_{t=0}^\infty (\delta_t+\eta \cdot \mathbbm{1}[a \neq \idle] - L_\eta^*)\big| \delta_0=\delta, r_0=r}.
\]
We also introduce the \textit{state-action cost function}
defined as
\begin{align}
Q_{\eta}(\delta,r,a)\triangleq \delta+\eta \cdot \mathbbm{1}[a \neq \idle]+\Exp{h_{\eta}(\delta',r')}
\label{eq:Bellman2}
\end{align}
for all $(\delta,r)\in \mathcal{S}, a \in \mathcal {A}$.
Then, also implied by Proposition~3.2 of \cite{sennott_1993}, the optimal deterministic policy for the Lagrangian problem with a given $\eta$ takes, for any $(\delta,r) \in \mathcal{S}$,  the action achieving the minimum in \eqref{eq:Bellman2}:
\begin{align}
\label{eq:opt_eta}
\pi_{\eta}^*(\delta,r) &\in \argmin_{a\in\{\idle,\new,\retx\}} Q_{\eta}(\delta,r,a)~. 
\end{align}

Focusing on deterministic policies, it is possible to characterize optimal policies for our CMDP problem: Based on Theorem~2.5 of \cite{sennott_1993}, we can prove the the following result:  


\begin{theorem}
There exists an optimal stationary policy for the CMDP in Problem~\ref{problem} that is optimal for the unconstrained problem considered in \eqref{eq:lag} for some $\eta=\eta^*$, and randomizes in at most one state. This policy can be expressed as a mixture of two deterministic policies $\pi^*_{\eta^*,1}$ and $\pi^*_{\eta^*,2}$ that differ in at most a single state $s$, and are both optimal for the Lagrangian problem \eqref{eq:lag} with $\eta=\eta^*$. More precisely, there exists  $\mu \in [0,1]$ such that the mixture policy $\pi^*_{\eta^*}$, which selects, in state $s$, $\pi^*_{\eta^*,1}(s)$ with probability $\mu$ and $\pi^*_{\eta^*,2}(s)$ with probability $1-\mu$, and otherwise follows these two policies (which agree in all other states)
is optimal for Problem~\ref{problem}, and the constraint in \eqref{eq:constraint} is satisfied with equality.   
\label{thm_mixture}
\end{theorem}
\begin{proof}
By Proposition~\ref{lemma_assump} in Appendix~\ref{Append_assumptions}, Theorem~2.5, Proposition~3.2, and Lemma~3.9 of~\cite{sennott_1993} hold for Problem~\ref{problem}. By
Theorem~2.5 of~\cite{sennott_1993}, there exists an optimal stationary policy that is a mixture of two deterministic policies, $\pi^*_{\eta^*,1}$ and $\pi^*_{\eta^*,2}$,  which differ in at most one state and are $\eta^*$-optimal by Proposition~3.2 of \cite{sennott_1993} satisfying \eqref{eq:Bellman} and \eqref{eq:Bellman2}. From Lemma 3.9 of~\cite{sennott_1993}, the mixture policy $\pi^*_{\mu}$, for any $\mu \in [0,1]$, also satisfies \eqref{eq:Bellman} and \eqref{eq:Bellman2}, and is optimal for the unconstrained problem in \eqref{eq:lag} with $\eta=\eta^*$. From the proof of Theorem~2.5 of~\cite{sennott_1993}, there exists a $\mu \in [0,1]$ such that $\pi^*_{\eta^*}$ satisfies the constraint in \eqref{eq:constraint} with equality. 
This completes the proof of the theorem.
\end{proof}

Some other results in \cite{sennott_1993} will be useful in determining $\pi^*_{\eta^*}$. For any $\eta>0$, let $C_{\eta}$ and $J_{\eta}$ denote the average number of transmissions and average AoI, respectively, for the optimal policy $\pi_{\eta}^*$. Note that these are multivalued functions since there might be more than one optimal policy for a given $\eta$. 
Note also that, $C_\eta$ and $J_\eta$ can be computed directly by finding the stationary distribution of the chain, or estimated empirically by running the MDP with policy $\pi^*_\eta$. From Lemma~3.4 of \cite{sennott_1993}, $L^*_{\eta}$, $C_\eta$ and $J_\eta$ are monotone functions of $\eta$:
if $\eta_1 < \eta_2$, we have $C_{\eta_1} \ge C_{\eta_2}$, $J_{\eta_1} \le J_{\eta_2}$ and $L^*_{\eta_1} \le L^*_{\eta_2}$. 
This statement is also intuitive since $\eta$ effectively represents the cost of a single transmission in \eqref{eq:Bellman} and \eqref{eq:Bellman2}, as $\eta$ increases, the average number of transmissions of the optimal policy cannot increase, and as a result, the AoI cannot decrease.

To determine the optimal policy, one needs to find $\eta^*$, the policies $\pi^*_{\eta^*,1}$ and $\pi^*_{\eta^*,2}$, and the weight $\mu$.
In fact, \cite{sennott_1993} shows that $\eta^*$ is defined as 
\begin{align}
\eta^* \triangleq \inf\{\eta>0:C_{\eta}\le C_{max}\},
\label{eq:eta_star}
\end{align}
where the inequality $C_{\eta}\le C_{max}$ is satisfied if it is satisfied for at least one value of (multivalued) $C_{\eta}$. By Lemma~3.12 of \cite{sennott_1993} and Proposition~\ref{lemma_assump}, $\eta^*$ is finite, and $\eta^*>0$ if $C_{max}<1$.

If $C^{\pi^*_{\eta^*,i}}=C_{max}$ for $i=1$ or $i=2$, then it is the optimal policy, that is, $\pi^*_{\mu}=\pi^*_{\eta^*,i}$ and $\mu=1$ if $i=1$ and $0$ if $i=2$.
Otherwise one needs to select $\mu$ such that $C^{\pi^*_\mu}=C_{max}$: that is, if $C^{\pi^*_{\eta^*,2}} < C_{max} < C^{\pi^*_{\eta^*,1}}$, then
\begin{equation}
\label{eq:random}
\mu = \frac{C_{max}-C^{\pi^*_{\eta^*,2}} }{C^{\pi^*_{\eta^*,1}}-C^{\pi^*_{\eta^*,2}}},  
\end{equation}
which results in an optimal policy.

In practice, finding both $\eta^*$ and the policies $\pi^*_{\eta^*,1}$ and $\pi^*_{\eta^*,2}$ is hard. However, given two monotone sequences sequences $\eta_n \uparrow \eta^*$ and $\eta'_n \downarrow \eta^*$, there is a subsequence of $\eta_n$ (resp., $\eta'_n$) such that the corresponding subsequence of the $\eta_n$-optimal policies $\pi^*_{\eta_n}$ ($\eta'_n$-optimal policies $\pi^*_{\eta'_n}$, resp.) satisfying the Bellman equation \eqref{eq:Bellman} converge. Then the limit points $\pi$ and $\pi'$\footnote{$\pi_n \to \pi$ if for any state $s$, $\pi_n(s)=\pi(s)$ for $n$ large enough.} are $\eta^*$-optimal by Lemma~3.7 (iii) of \cite{sennott_1993} and $C^{\pi} \ge C_{max} \ge C^{\pi'}$ by the monotonicity of $C_{\eta}$ and the same Lemma~3.7. Although there is no guarantee that $\pi$ and $\pi'$ only differ in a single point, we can combine them to get an optimal randomized policy using $\mu$ defined in \eqref{eq:random}. In this case, Lemma~3.9 of \cite{sennott_1993} implies that the policy that first randomly selects if it should use $\pi$ or $\pi'$ (choosing $\pi$ with probability $\mu$) and then uses the selected policy forever is $\eta^*$-optimal. However, since $(1,0)$ is a positive recurrent state of both policies and they have a single recurrent class by Proposition~3.2 of \cite{sennott_1993}, we can do the random selection of between $\pi$ and $\pi'$ independently every time the system gets into state $(1,0)$ without changing the long-term average or expected AoI and transmission cost (note that one cannot choose randomly between the two policies in, e.g., every step). Thus, the resulting randomized policy is $\eta^*$-optimal, and since $\mu$ is selected in such a way that the total transmission cost is $C_{max}$, it is also an optimal solution of Problem~\ref{problem} by Lemma~3.10 of \cite{sennott_1993}. Note that to derive two $\eta^*$-optimal policies, which provably differ only in a single state, a much more elaborate construction is used in  \cite{sennott_1993}. However, in practice, $\pi$ and $\pi'$ obtained above are often the same except for a single state. Furthermore, we can approximate $\pi_1$ and $\pi_2$ by $\pi^*_{\eta_n}$ and $\pi^*_{\eta'_n}$ for $n$ large enough. This idea is explored in the next section.

Theorem~\ref{thm_mixture} and the succeeding discussion present the general structure of the optimal policy. In Section~\ref{sec:primal}, for practical implementation,  a computationally efficient heuristic algorithm is proposed based upon the discussion in this section.




\section{An Iterative Algorithm to Minimize the AoI under an Average Cost Constraint}
\label{sec:primal}

 
While our state space is countably infinite, since the age can be arbitrarily large ($r_{max}$ may  also be infinite), in practice we can approximate the countable state space with a large but finite space by setting an upper bound on the age (which will be denoted by $N$), and by selecting a finite $r_{max}$ (whenever the chain would leave this constrained state space, we truncate the value of the age and/or the retransmission number to $N$ and $r_{max}$, respectively); this gives a finite state space approximation to the problem similarly to \cite{Altman2010,hsuage2017}. Clearly, letting $N$ and $r_{max}$ go to infinity, the optimal policy for the restricted state space will converge to that of the original problem.

When we consider the finite state space approximation of our problem, we can employ the \emph{relative value iteration} (RVI) \cite{Puterman_book} algorithm to solve \eqref{eq:Bellman} for any given $\eta$, and hence find (an approximation of) the optimal policy $\pi^*_\eta$. Note that the finite state space approximation is needed for the practical implementation of the RVI algorithm since each iteration in RVI requires the computation of the value function for each state-action pair (for the infinite state space we would need to use some sort of parametric approximation of the states or the value functions, which is out of the scope of this paper). The pseudocode of the RVI algorithm is given in Algorithm~\ref{alg:RVI}.  To simplify the notation, the dependence on $\eta$ is suppressed in the algorithm for $h, V$ and $Q$.

\begin{algorithm}
\begin{small}
\label{alg:RVI}
\DontPrintSemicolon
\SetAlgoLined
\SetKwInOut{Input}{Input}\SetKwInOut{Output}{Output}
\Input{Lagrange parameter $\eta$, error probability  $g(r)$}
\BlankLine
  $(\delta^{ref},r^{ref})$ \tcc{choose an arbitrary but fixed reference state}
$n\leftarrow 0~$    \tcc{iteration counter}

$h_0^{N\times r_{max}}\leftarrow \mathbf{0}~$   \tcc{initialization} 

\BlankLine 
\While{$1$ \tcc{until convergence} }  {    
    
    \For{state $s=(\delta,r)\in [1,\ldots,N]\times [1,\ldots,r_{max}]$ }{   
    
    \For{action $a\in\mathcal{A}$}{    
    
	$Q_{n+1}(\delta,r,a)\leftarrow\delta+\eta\cdot\mathbbm{1}[a^{\pi}\neq \idle]+	\Exp{h_n(\delta',r')}$\;  
    }

${V}_{n+1}(\delta,r)\leftarrow\min_{a}(Q_{n+1}(\delta,r,a))$\;
$h_{n+1}(\delta,r)\leftarrow{V}_{n+1}(\delta,r)-{V}_{n+1}(\delta^{ref},r^{ref})$\;
    
    }
    
    \eIf{$|h_{n+1}-h_n|\leq \epsilon$}{
\tcc{compute the optimal policy}

\For{$(\delta ,r) \in [1,\ldots,N]\times [1,\ldots,r_{max}]$}{    
	
	 $\pi^*_{\eta}(\delta,r)\leftarrow\argmin_{a}(Q(\delta,r,a))$ 
    
    }
    \Return $\pi^*$ \; 
}{
increase the iteration counter:
       $n\leftarrow n+1$ \;
}


    
	
    
          
}
\caption{Relative value iteration (RVI) algorithm for a given $\eta$.}
\end{small}
\end{algorithm}

After presenting an algorithm that can compute the optimal deterministic policy $\pi^*_\eta$ for any given $\eta$ (more precisely, an arbitrarily close approximation thereof in the finite approximate MDP), we need to find the particular Lagrange multiplier $\eta^*$ as defined by \eqref{eq:eta_star}.
As the simplest solution, we would need to generate $C_{\eta}$ for a reasonably large range of $\eta$ values to determine $\eta^*$. This could be approximated by computing $C_\eta$ for a fine grid of $\eta$ values, but this approach might be computationally demanding (note that generating each point requires running an instance of RVI). 

Instead, we can use the following heuristic: 
With the aim of finding a single $\eta$ value with $C_\eta\approx C_{max}$, we start with an initial parameter $\eta^0$, and run an iterative algorithm updating $\eta$ as $\eta^{m+1} = \eta^m+\alpha_m (C_{\eta^m}-C_{max})$ for a step size parameter $\alpha_m$\footnote{$\alpha_m$ is a positive decreasing sequence and satisfies the following conditions: $\sum_m \alpha_m = \infty$ and $\sum_m \alpha^2_m < \infty$ from the theory of stochastic approximation \cite{stochastic_approx}.} (note that for each step we need to run RVI to be able to determine $C_{\eta^m}$). We continue this iteration until $|C_{\eta^{m}}-C_{max}|$ becomes smaller than a given threshold, and denote the resulting value by ${\eta}^*$. We can increase or decrease the $\eta^*$ value until $\eta^*$ and its modification satisfy the conditions (note that with a finite state space, which is an approximation always used when computing an optimal policy numerically, $\pi_\eta$, and consequently $C_\eta$ and $J_\eta$, are piecewise constant functions of $\eta$, thus the step size must be chosen sufficiently large to change the average transmission cost). 

In order to obtain two deterministic policies and the corresponding mixing coefficient, based on the discussion at the end of Section~\ref{sec:structure}, we want to find optimal policies for $\eta$ values slightly smaller and larger than $\eta^*$, and so we compute the optimal policies (by RVI) for $\eta^*\pm \xi$  where $\xi$ is a small perturbation and obtain a mixture coefficient according to \eqref{eq:random} as
\begin{equation}
\mu=\frac{C_{max}-C_{\eta^*+ \xi}}{C_{\eta^*- \xi}-C_{\eta^*+ \xi}}~.
\end{equation}
If the optimal policies differ only in a single state, we can randomize in that state (by Theorem~\ref{thm_mixture}), while, if they are more different, we can randomly select between the policies (with probabilities $\mu$ and $1-\mu$) every time after a successful transmission (i.e., when the system is in state $(1,0)$), as discussed at the end of Section~\ref{sec:structure}.

Numerical results obtained by implementing the above heuristics in order to minimize the average AoI with HARQ will be presented in Section~\ref{sec:results}. In the next section, we focus on the simpler scenario with the classical ARQ protocol.

\section{AoI with Classical ARQ Protocol under an Average Cost Constraint}
\label{sec:arq}

In the classical ARQ protocol, failed transmissions are discarded at the destination and the receiver tries to decode each retransmission as a new message. In the context of AoI, there is no point in retransmitting an undecoded packet since the probability of a successful transmission is the same for a retransmission and for the transmission of a new update.  Hence, the state space reduces to  $\delta \in \{1,2,\ldots\}$ as $r_t=0$ for all $t$, and the action space reduces to $\mathcal{A}\in\{\new , \idle\}$, and the probability of error $p\triangleq g(0)$ is fixed for every transmission attempt.\footnote{This simplified model with classical ARQ protocol and Lagrangian relaxation is equivalent to the work in \cite{Altman2010} when $\eta$ is considered to be the cost of a single transmission and the assumption of a perfect transmission channel in \cite{Altman2010} is ignored.}  State transitions in (\ref{eq:transitions}), Bellman optimality equations \cite{Puterman_book,Bertsekas2000} for the countable-state MDP in (\ref{eq:Bellman}), and the RVI algorithm with the finite state approximation can all be simplified accordingly. We define
\begin{align}
Q_{\eta}(\delta,\idle)&\triangleq \delta+ h_{\eta}(\delta+1),\label{eq:update} \\
Q_{\eta}(\delta,\new)&\triangleq \delta+\eta+p h_{\eta}(\delta+1)+(1-p)h_{\eta}(1), \label{eq:update1}
\end{align}
where $h_{\eta}(\delta)$  is the optimal differential value function satisfying the Bellman optimality equation
\begin{align}
\label{eq:update2}
h_{\eta}(\delta)+L^*_{\eta} \triangleq \min{\{Q_{\eta}(\delta,\idle),Q_{\eta}(\delta,\new)\}}, ~\forall \delta \in \{1,2,\ldots \}. 
\end{align}
Thanks to these simplifications, we are able to provide a closed-form solution to the corresponding Bellman equations in \eqref{eq:update}, \eqref{eq:update1} and \eqref{eq:update2}. 
\begin{lemma}
\label{lem:ARQ}
The policy that satisfies the Bellman optimality equations for the standard ARQ protocol is deterministic and has a threshold structure:
\begin{align*}
\pi^*(\delta)=\begin{cases} 
\new &\mbox{if} ~\delta \geq \Delta_\eta,   \\
\idle &\mbox{if} ~\delta < \Delta_\eta.
\end{cases}
\end{align*}
for some integer $\Delta_\eta$ that depends on $\eta$.
\end{lemma}

\begin{proof}
The proof is given in Appendix~\ref{AppendA}. 
\end{proof}

The next lemma characterizes the possible values of the threshold defined in Lemma~\ref{lem:ARQ}.
\begin{lemma}
\label{thm1}
Under the standard ARQ protocol, the $\eta$-optimal value of the threshold $\Delta_\eta$ can be found in closed-form:
\begin{align*}
\Delta^*_\eta \in \left\{ \left\lfloor\frac{\sqrt{2\eta(1-p)+p}-p}{1-p} \right\rfloor, \left\lceil\frac{\sqrt{2\eta(1-p)+p}-p}{1-p} \right\rceil  \right\}.
\end{align*}
\end{lemma}
\begin{proof}
The proof is given in Appendix~\ref{AppendB}. 
\end{proof}

The main result of this section, given below, shows that the optimal policy for Problem~\ref{problem} is a randomized threshold policy which randomizes over the above two thresholds for the optimal value of $\eta^*$.
Let $\Delta_{C_{max} } \triangleq \frac{1/C_{max} -p}{1-p}$, $\Delta_1\triangleq\lfloor \Delta_{C_{max} } \rfloor$ and $\Delta_2\triangleq\lceil \Delta_{C_{max} } \rceil$, and consider the mixture of the threshold policies with thresholds  $\Delta_1$ and $\Delta_2$, respectively, and mixture coefficient $\mu \in [0,1]$.
The resulting policy $\pi^*_{C_{max},\mu}$ can be written in closed form:
if $\Delta_{C_{max} }$ is an integer then $\pi^*_{C_{max},\mu}(\delta)=\new$ if $\delta\ge \Delta_{C_{max} }$ and $\idle$ otherwise. If
$\Delta_{C_{max}}$ is not an integer, then $\pi^*_{C_{max},\mu}(\delta) = \new$  if  $\delta \ge \lceil \Delta_{C_{max} } \rceil$, 
$\pi^*_{C_{max},\mu}(\delta) = \idle$ if $\delta < \lfloor \Delta_{C_{max} } \rfloor$, while
$\pi^*_{C_{max},\mu}(\new|\delta) = \mu$ and $\pi^*_{C_{max},\mu}(\idle|\delta) = 1-\mu$ for $\delta =  \lfloor \Delta_{C_{max} } \rfloor$.
The mixture coefficient $\mu$ is selected so that $C^{\pi^*_{C_{max},\mu}} = C_{max}$:
From the proof of Lemma~\ref{thm1} one can easily deduce that the transmission cost (per time slot) of the threshold policy for any integer threshold $\Delta$ is given by
\begin{align}
C^{\Delta} = \frac{1}{\Delta (1-p)+p}.
\label{eq:const}
\end{align}
Hence, selecting  $\mu^* = \frac{C_{max} - C^{\Delta_2}}{C^{\Delta_1}-C^{\Delta_2}}$, as described in \eqref{eq:random}, ensures 
$C^{\pi^*_{C_{max},\mu^*}} = C_{max}$. Denoting $\pi^*_{C_{max}}=\pi^*_{C_{max},\mu^*}$, we obtain the following theorem (the proof is given in Appendix~\ref{AppendC}).

\begin{theorem}
\label{thm:ARQ}
For any $C_{max} \in (0,1]$, the stationary policy $\pi^*_{C_{max},\mu^*}$ defined above is an optimal policy (i.e., a solution of Problem~\ref{problem}) under the ARQ protocol.
\end{theorem}

Numerical results obtained for the above algorithm will be presented and compared with  those from the HARQ protocol in Section~\ref{sec:results}.

\section{Learning to minimize AoI in an unknown environment}
\label{sec:learning}

\begin{algorithm}[t]
\begin{small}
\label{algo_learn}
\DontPrintSemicolon
\SetAlgoLined
\SetKwInOut{Input}{Input}\SetKwInOut{Output}{Output}
\Input{Lagrange parameter $\eta$ \tcc{error probability $g(r)$ is unknown}} 

$n\leftarrow 0~$    \tcc{time iteration}
$\tau\leftarrow1$     \tcc{softmax temperature parameter}  
$Q_{\eta}^{N\times M\times 3}\leftarrow 0~$   \tcc{initialization of $Q$} 
$L_{\eta}\leftarrow 0~$   \tcc{initialization of the gain} 
\BlankLine
\For{$n$}{    
    \BlankLine  
     \textsc{observe} the current state $s_n$ \\
    \For{$a \in \mathcal{A}$}{   
    
    \tcc{since it is a minimization problem, use minus $Q$ function in softmax} 
    
    $\pi(a|s_n)=\frac{\displaystyle\exp(-Q_{\eta}(s_n,a)/\tau)}{\displaystyle\sum_{a'\in\mathcal{A}}{\exp(-Q_{\eta}(s_n,a')/\tau)}}$ 
    }    
    \textsc{Sample} $a_n$ from $\pi(a|S_n)$ \;
    \textsc{observe} the next state $s_{n+1}$ and cost $c_n=\delta_n+\eta 1_{\{a_n=1,2\}}$
    
	\For{$a \in \mathcal{A}$}{    
	
	\tcc{softmax is also used for the next state $s_{n+1}$, so that it is on-policy }
    
    $\pi(a|s_{n+1})=\frac{\displaystyle\exp(-Q_{\eta}(s_{n+1},a_{n+1})/\tau)}{\displaystyle\sum_{a'_{n+1}\in\mathcal{A}}{\exp(-Q_{\eta}(s_{n+1},a'_{n+1})/\tau)}}$ 
    
    }    
	
	\textsc{Sample} $a_{n+1}$ from $\pi(a_{n+1}|s_{n+1})$\;
    
    \textsc{Update} 
    
    $\alpha_n\leftarrow1/\sqrt{n}$  ~\tcc{update parameter}

    $Q_{\eta}(s_n,a_n)\leftarrow Q_{\eta}(s_n,a_n) + \alpha_n [\delta+\eta\cdot\mathbbm{1}[a_n\neq \idle]-J_{\eta}+Q_{\eta}(s_{n+1},a_{n+1})-Q_{\eta}(s_n,a_n)] $
    
            $L_{\eta}\leftarrow L_{\eta}+ 1/n [\delta+\eta\cdot\mathbbm{1}[a_n\neq \idle]-J_{\eta}]$ ~\tcc{update $J_{\eta}$ at every step} 
     
       $n\leftarrow n+1$ ~\tcc{increase the iteration} 
 
}
\caption{Average-cost SARSA with softmax}
\end{small}
\end{algorithm}

In the CMDP formulation presented in Sections~\ref{sec:primal} and~\ref{sec:arq}, we have assumed that the channel error probabilities for all retransmissions are known in advance. However, in most practical scenarios, these error probabilities may not be known at the time of deployment, or may change over time. Therefore, in this section, we assume that the source node does not have \textit{a priori} information about the decoding error probabilities, and has to learn them. We employ an online learning algorithm to learn $g(r)$ over time without degrading the performance significantly. 

The literature for average-cost RL is quite limited compared to discounted cost problems \cite{Mahadevan1996}, \cite{Sutton1998}. SARSA \cite{Sutton1998} is a well-known RL algorithm, originally proposed for discounted MDPs, that learns the optimal policy for an MDP based on the action performed by the current policy in a recursive manner.  For average AoI minimization in Problem~\ref{problem}, an average cost version of the SARSA algorithm is employed with \emph{Boltzmann}  (\emph{softmax}) exploration. The resulting algorithm is called \emph{average-cost SARSA with softmax}. 

As indicated by \eqref{eq:Bellman} and \eqref{eq:Bellman2} in Section \ref{sec:structure}, $Q_{\eta}(s_n,a_n)$ of the current state-action pair can be represented in terms of the immediate cost of the current state-action pair and the differential state-value function $h_{\eta}(s_{n+1})$ of the next state. Notice that, one can select the optimal actions by only knowing $Q_{\eta}(s,a)$ and choosing the action that will give the minimum expected cost as in \eqref{eq:opt_eta}.  Thus, by only knowing $Q_{\eta}(s,a)$, one can find the optimal policy $\pi^*$ without knowing the transition probabilities $\Pt$ characterized by $g(r)$ in \eqref{eq:transitions}. 

Similarly to SARSA, \emph{average-cost SARSA with softmax} starts with an initial estimation of $Q_{\eta}(s,a)$ and finds the optimal policy by estimating state-action values in a recursive manner. In the $n^{th}$ time iteration, after taking action $a_n$, the source observes the next state $s_{n+1}$, and the instantaneous cost value $c_n$. Based on this, the estimate of $Q_{\eta}(s,a)$ is updated by weighing the previous estimate and the estimated expected value of the current policy in the next state $s_{n+1}$.  Also note that, in general, $c_n$ is not necessarily known before taking action $a_n$ because it does not know the next state $s_{n+1}$ in advance. In our problem, the instantaneous cost $c_n$ is the sum of AoI at the destination and the cost of transmission, i.e. $\delta_n+\eta\cdot\mathbbm{1}[a_n\neq \idle]$; hence, it is readily known at the source node.

In each time slot, the learning algorithm 
\begin{itemize}
\item observes the current state $s_n \in \mathcal{S}$,
\item selects and performs an action $a_n \in \mathcal{A}$,
\item observes the next state $s_{n+1}\in \mathcal{S}$ and the instantaneous cost $c_n$,
\item updates its estimate of $Q_{\eta}(s_n,a_n)$ using the current estimate of $_{\eta}$ by
\begin{equation}
Q_{\eta}(s_n,a_n)\leftarrow Q_{\eta}(s_n,a_n) + \alpha_n [\delta+\eta\cdot\mathbbm{1}[a_n\neq \idle]-L_{\eta}+Q_{\eta}(s_{n+1},a_{n+1})-Q_{\eta}(s_n,a_n)],
\end{equation}
where $\alpha_n$ is the update parameter (learning rate) in the $n^{th}$ iteration.
\item updates its estimate of $L_{\eta}$ based on empirical average.
\end{itemize}
The details of the algorithm are given in Algorithm \ref{algo_learn}. 
We update the gain $L_\eta$ at every time slot based on the empirical average, instead of updating it at non-explored time slots. 

As we discussed earlier, with the accurate estimate of $Q_{\eta}(s,a)$ at hand the transmitter can decide for the optimal actions for a given $\eta$ as in \eqref{eq:opt_eta}. However, until the state-action cost function is accurately estimated, the transmitter action selection method should balance the \textit{exploration} of new actions with the \textit{exploitation} of actions known to perform well. In particular, the \textit{Boltzmann} action selection method, which chooses each action probabilistically relative to expected costs, is used in this paper. The source
assigns a probability to each action for a given state $s_n$, denoted by $\pi(a|s_n)$:
\begin{equation}
\pi(a|s_n)\triangleq \frac{\displaystyle\exp(-Q_{\eta}(s_n,a)/\tau)}{\displaystyle\sum_{a'\in\mathcal{A}}{\exp(-Q_{\eta}(s_n,a')/\tau)}},
\end{equation}
where $\tau$ is called the temperature parameter such that high $\tau$ corresponds to more uniform action selection (exploration) whereas low $\tau$ is biased toward the best action (exploitation).

In addition, the constrained structure of the average AoI problem requires additional modifications to the algorithm, which is achieved in this paper by updating the Lagrange multiplier according to the empirical resource consumption.  In each time slot, we keep track of a value $\eta$ resulting in a transmission cost close to $C_{max}$, and then find and apply a policy that is optimal (given the observations so far) for the MDP with Lagrangian cost as in Algorithm~\ref{algo_learn}. 

The performance of \emph{average-cost SARSA with softmax}, and its comparison with the RVI algorithm will be presented in the next section.

\section{Numerical Results}
\label{sec:results}

In this section, we provide numerical results for all the proposed algorithms, and compare the achieved average performances. For the simulations employing HARQ, 
motivated by previous research on HARQ \cite{hybrid2001}, \cite{harq2003}, \cite{Lagrange2010}, we assume that decoding error reduces exponentially with the number of retransmission, that is, $g(r)\triangleq p_0 \lambda^{r}$ for some $\lambda \in (0,1)$, where $p_0$ denotes the error probability of the first transmission, and $r$ is the retransmission count (set to $0$ for the first transmission). 
The exact value of the rate $\lambda$ depends on the particular HARQ protocol and the channel model. Note that ARQ corresponds to the case with $\lambda=1$ and $r_{max}= 0$. Following the \emph{IEEE 802.16} standard\cite{IEEEstandard}, the maximum number of retransmissions is set to $r_{max}=3$; however, we will present results for other $r_{max}$ values as well. We note that we have also run simulations for HARQ with relatively higher $r_{max}$ values and $r_{max}=\infty$, and the improvement on the performance is not observable beyond $r_{max}=3$. Numerical results for different $p_0$, $\lambda$ and  $C_{max}$ values, corresponding to different channel conditions and HARQ schemes, will also be provided.

Figure~\ref{fig:update} illustrates the deterministic policies obtained by RVI and the search for $\eta^*$ for given  $C_{max}$ and $p_0$ values, while $\lambda$ is set to $0.5$. The final policies are generated by randomizing between  $\pi^*_{\eta^*-\xi}$ and $\pi^*_{\eta^*+\xi}$; the approximate $\eta^*$ values found for the settings in Figures~\ref{fig:first} and~\ref{fig:second} are $5$ and $19$, respectively, and $\xi$ is set to $0.2$. As it can  be seen from the figures, the resulting policy transmits less as the average cost constraint becomes more limiting, i.e., as $\eta$ increases. We also note that, although the policies $\pi^*_{\eta^*-\xi}$ and $\pi^*_{\eta^*+\xi}$ are obtained for similar $\eta^*$ values, and hence, have similar average number of transmissions, they may act quite differently especially for large $C_{max}$ values.

\begin{figure}
\centering
\subfigure[$C_{max}=0.4$, $p_0=0.3$]{%
\label{fig:first}%
\includegraphics[scale=0.5]{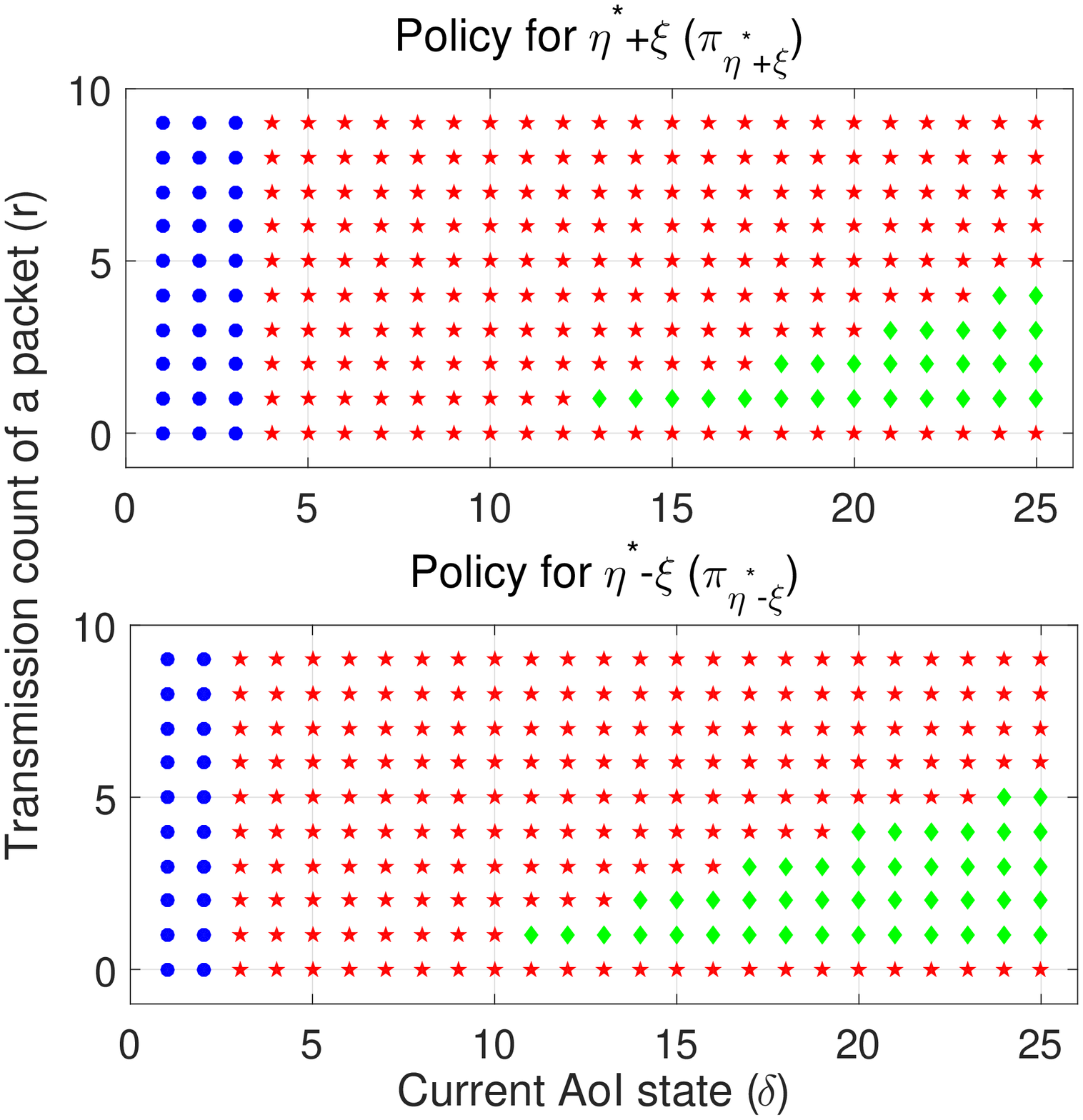}}%
\subfigure{\raisebox{40mm}{
\includegraphics[scale=0.2]{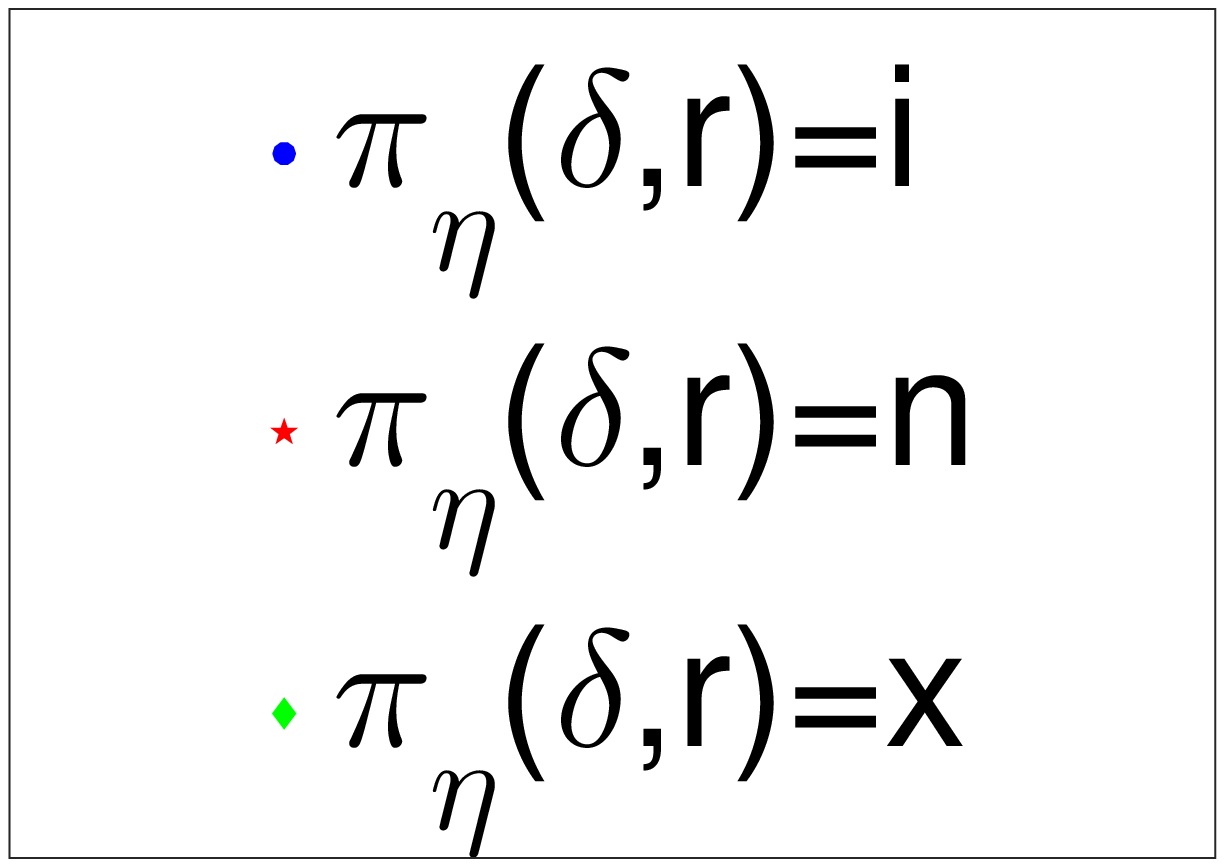}}}%
\qquad
\renewcommand{\thesubfigure}{(b)}
\subfigure[$C_{max}=0.2$, $p_0=0.4$]{%
\label{fig:second}%
\includegraphics[scale=0.5]{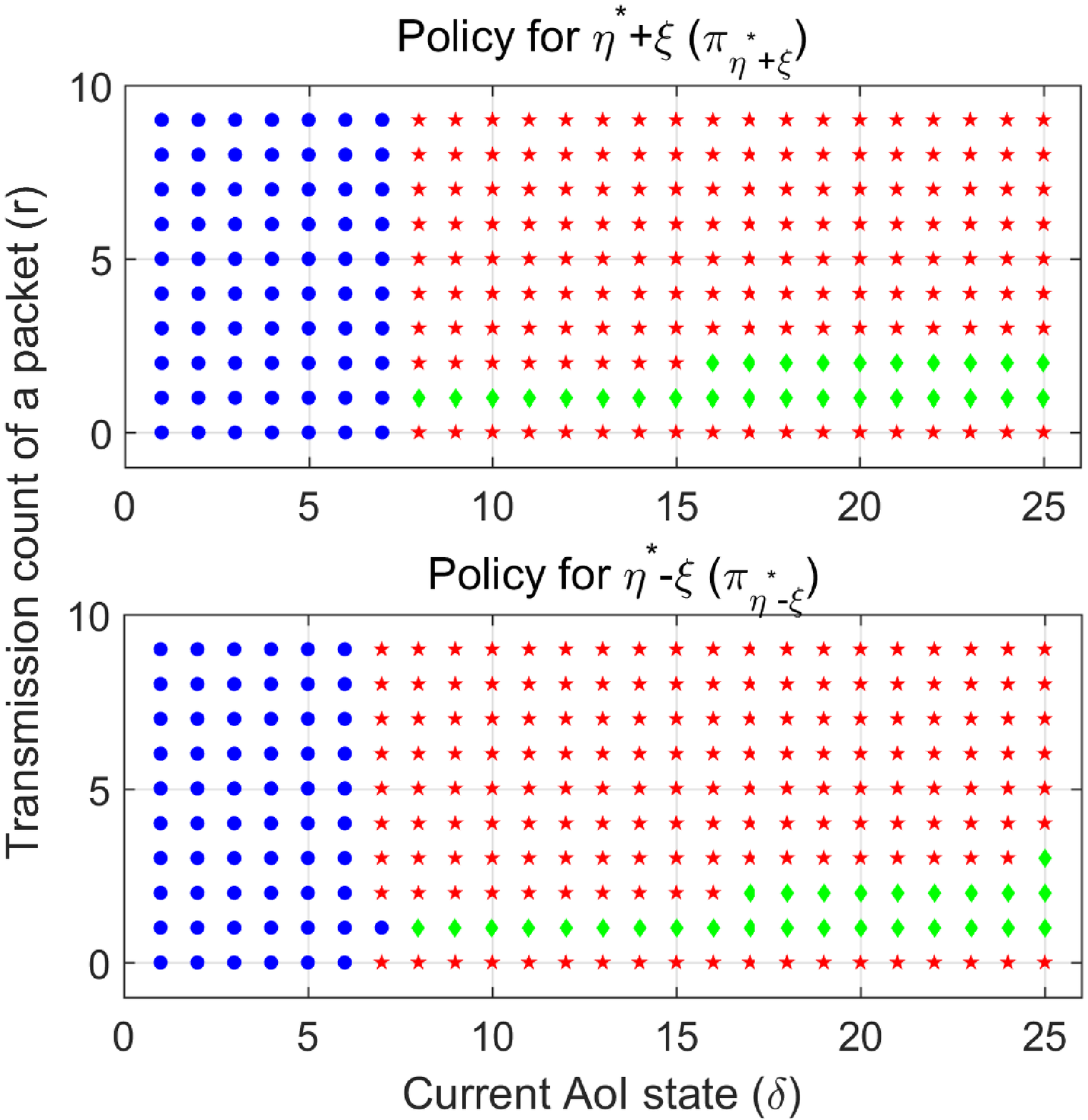}}%
\subfigure{\raisebox{40mm}{
\includegraphics[scale=0.2]{legenddd2}}}%
\caption{Deterministic policies $\pi_{\eta^*+\xi}$ (top) and $\pi_{\eta^*-\xi}$ (bottom)  when  $\lambda=0.5$ and $r_{max}=9$. (Blue circles, red stars, and green diamonds represent actions $\pi_{\eta}(\delta,r)= \idle$, $\new$ and $\retx$, respectively.)}
\label{fig:update}
\end{figure}

Figure~\ref{fig:detervsrandom_hybrid} illustrates the performance of the proposed randomized HARQ policy with respect to $C_{max}$ for different $p_0$ values when $\lambda$ is set to $0.5$. We also include the performance of the optimal deterministic and randomized threshold policies with ARQ, derived in Section~\ref{sec:arq}, for $p_0 = 0.5$. For baseline, we use a simple no-feedback policy that periodically transmits a fresh status update with a period of $\lceil{1/C_{max}}\rceil$, ensuring that the constraint on the average number of transmissions holds. The effect of feedback on the performance can be seen immediately: a single-bit ACK/NACK feedback, even with the ARQ protocol, decreases the average AoI considerably, although receiving feedback might be costly for some status update systems. The two curves for the ARQ policies demonstrate the effect of randomization: the curve corresponding to the randomized policy is the lower convex hull of the piecewise constant AoI curve for deterministic policies. For the same $p_0=0.5$, HARQ with $\lambda=0.5$ improves only slightly over ARQ. Smaller $p_0$ results in a decrease in the average AoI as expected, and the gap between the AoIs for different $p_0$ values is almost constant for different $C_{max}$ values. 

More significant gains can be achieved from HARQ when the error probability decreases faster with retransmissions (i.e., small $\lambda$), or more retransmissions are allowed. This is shown in Figure~\ref{fig:arq_vs_harq}. On the other hand, the effect of retransmissions on the average AoI (with respect to ARQ) is more pronounced when $p_0$ is high and $\lambda$ is low.

\begin{figure}
\centering
\includegraphics[scale=0.75]{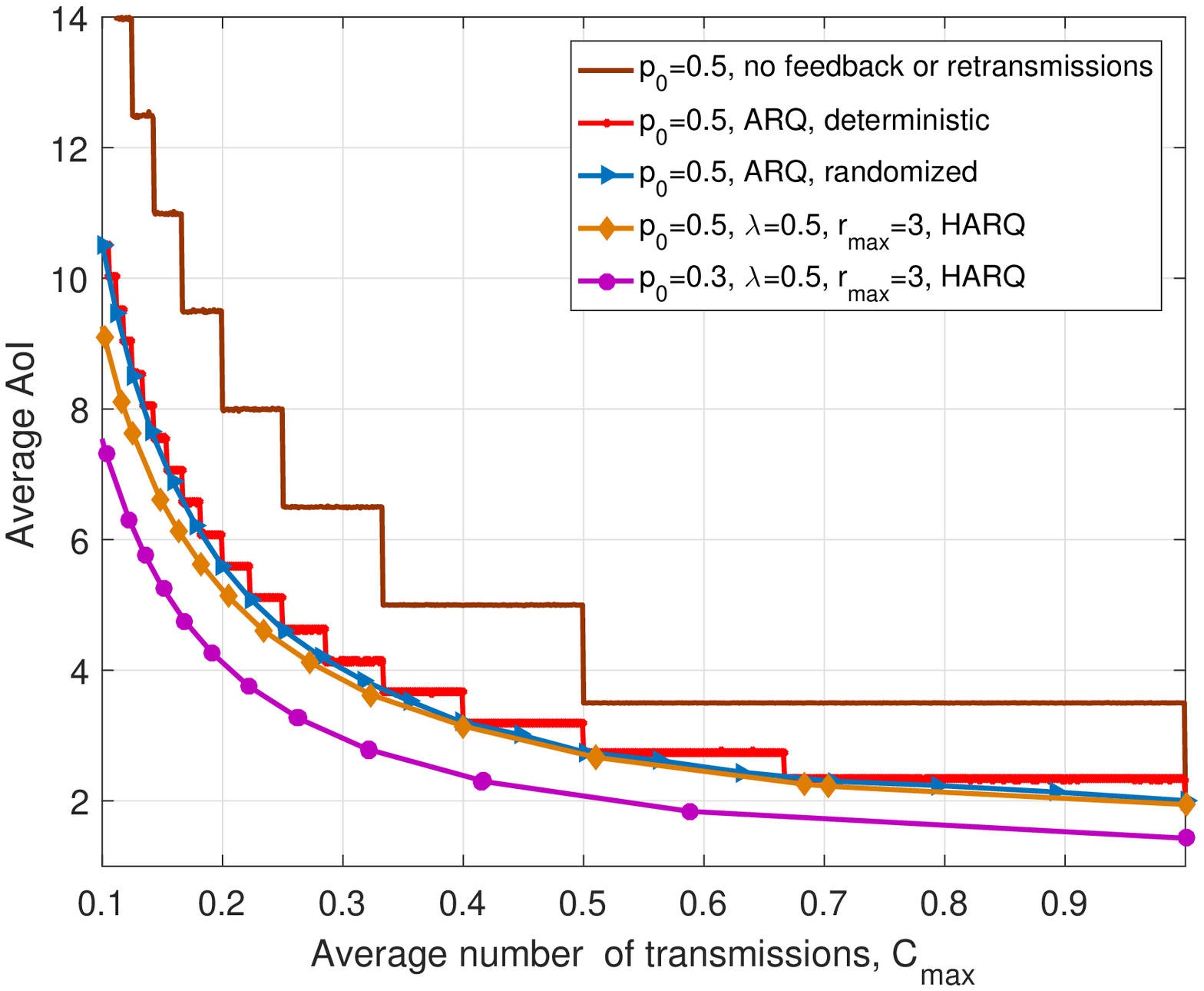}
\caption{Expected average AoI as a function of $C_{max}$  for ARQ and HARQ protocols for different $p_0$ values. Time horizon is set to $T=10000$, and the results are averaged over $1000$ runs.}
\label{fig:detervsrandom_hybrid}
\end{figure}

\begin{figure}
\centering
\includegraphics[scale=0.75]{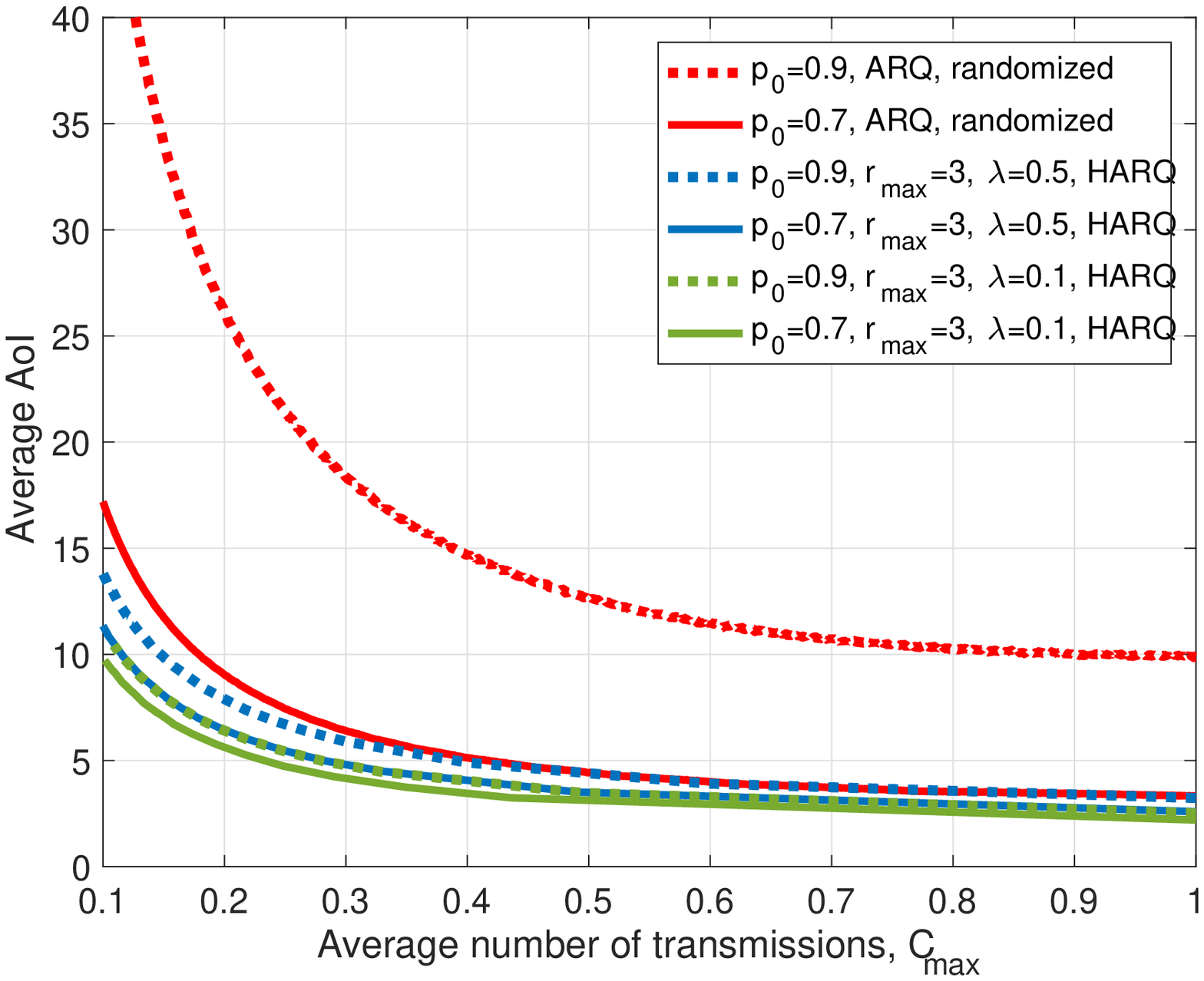}
\caption{Expected average AoI with respect to $C_{max}$  for ARQ and HARQ protocols for different $p_0$ and $r_{max}$ values. Time horizon is set to $T=10000$, and the results are averaged over $1000$ runs.}
\label{fig:arq_vs_harq}
\end{figure}


Figure~\ref{fig:wrt_x} shows the average AoI achieved by the HARQ protocol with respect to different $p_0$ and $\lambda$ values for $r_{max}=3$. 
Similarly to Figure~\ref{fig:detervsrandom_hybrid}, the gap between the average AoI values is higher for unreliable environments with higher error probability, and the performance gap due to different $\lambda$ values are not observable for relatively reliable environments, for example, when $p_0=0.3$. The performance difference for different $\lambda$ values (with a fixed $p_0$) is more pronounced when the average number of transmissions, $C_{max}$, is low, since then less resources are available to correct an unsuccessful transmission.

\begin{figure}
\centering
\includegraphics[scale=0.75]{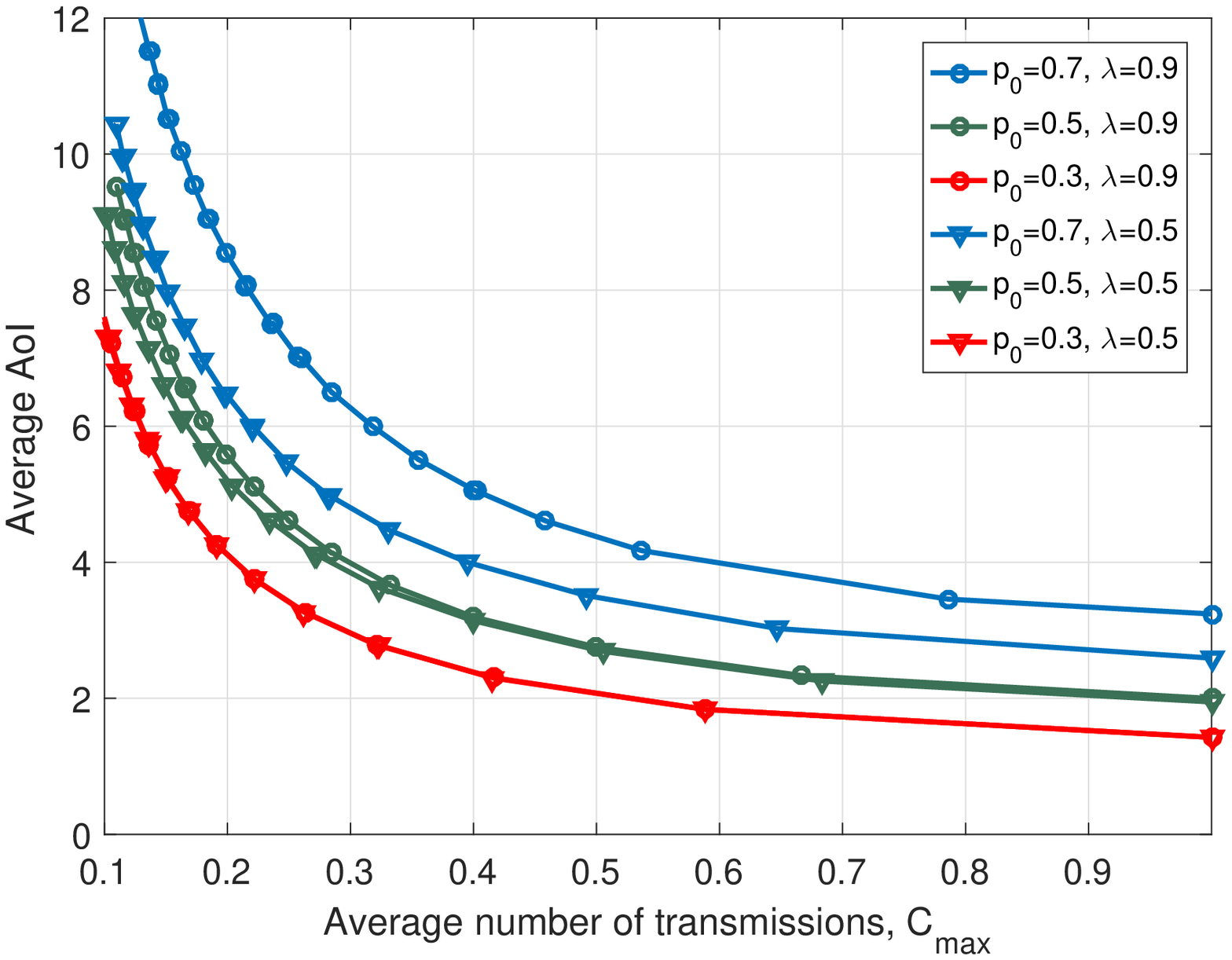}
\caption{Expected average AoI with respect to $C_{max}$  for HARQ protocols with different $g(r)=p_0 \lambda^{r}$ values corresponding to different $p_0$ and $\lambda$ values with  $r_{max}=3$. The time horizon is set to $T=10000$, and the results are averaged over $1000$ runs.}
\label{fig:wrt_x}
\end{figure}

Figure~\ref{fig:learn1} shows the evolution of the average AoI over time when the average-cost SARSA learning algorithm is employed. It can  be observed that the average AoI achieved by Algorithm~\ref{algo_learn}, denoted by \emph{RL} in the figure, converges to the one obtained from the RVI algorithm which has \emph{a priori} knowledge of $g(r)$. We can observe from Figure~\ref{fig:learn1} that the performance of SARSA achieves that of RVI in about $10000$ iterations. Figure~\ref{fig:learn2} shows the performance of the two algorithms (with again $10000$ iterations in SARSA) as a function of $C_{max}$ in two different setups. We can see that SARSA performs very close to RVI with a gap that is more or less constant for the whole range of $C_{max}$ values. We can also observe that the variance of the average AoI achieved by SARSA is much larger when the number of transmissions is limited, which also limits the algorithm's learning capability. 
\begin{figure}
\centering
\includegraphics[scale=0.75]{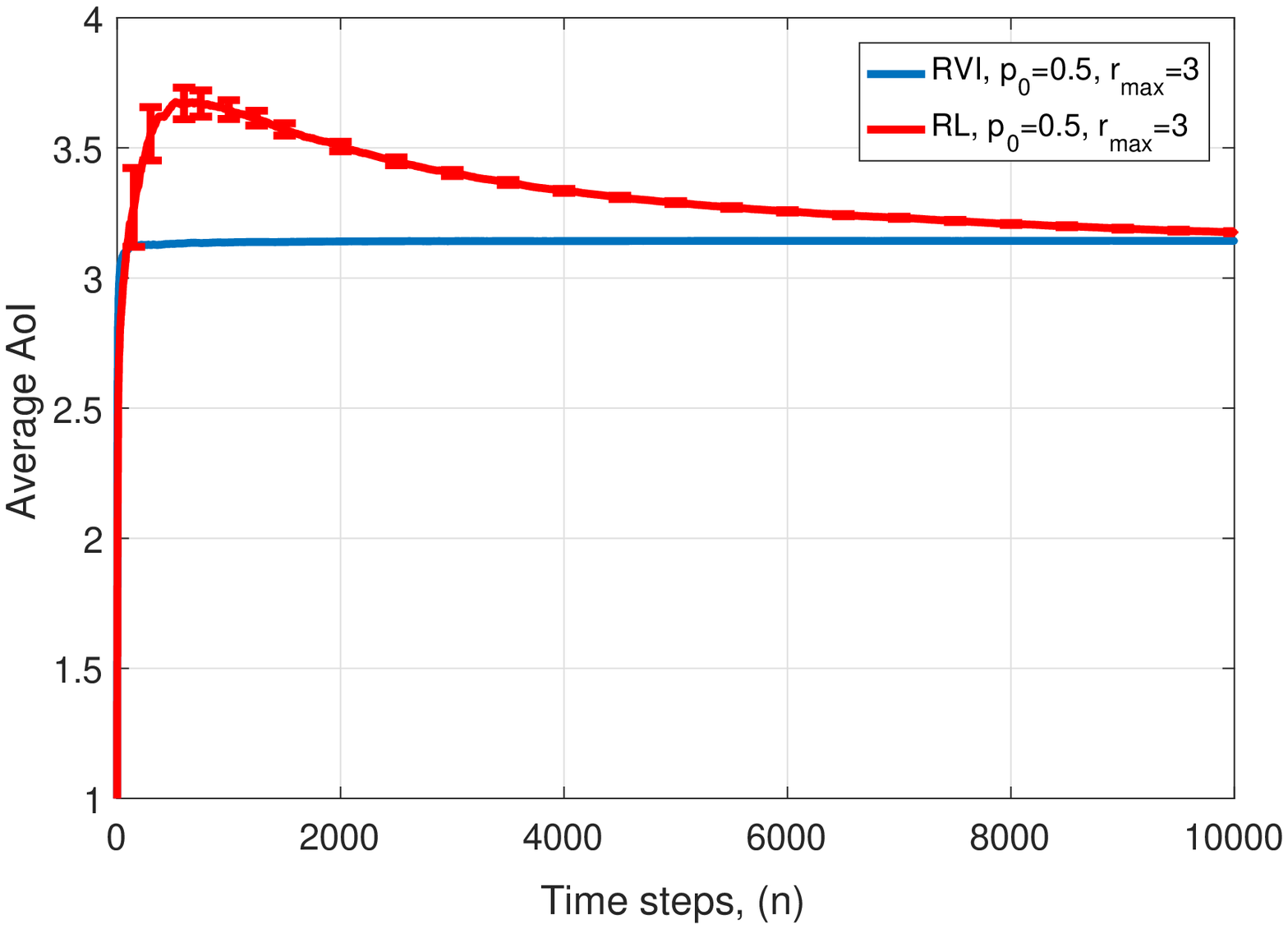}
\caption{Performance of the average-cost SARSA for $r_{max}=3$, $p_0=0.5$, $\lambda=0.5$, $C_{max}=0.4$ and $n= 10000$, averaged over 1000 runs (both the mean and the variance are shown).}
\label{fig:learn1}
\end{figure}
\begin{figure}
\centering
\includegraphics[scale=0.75]{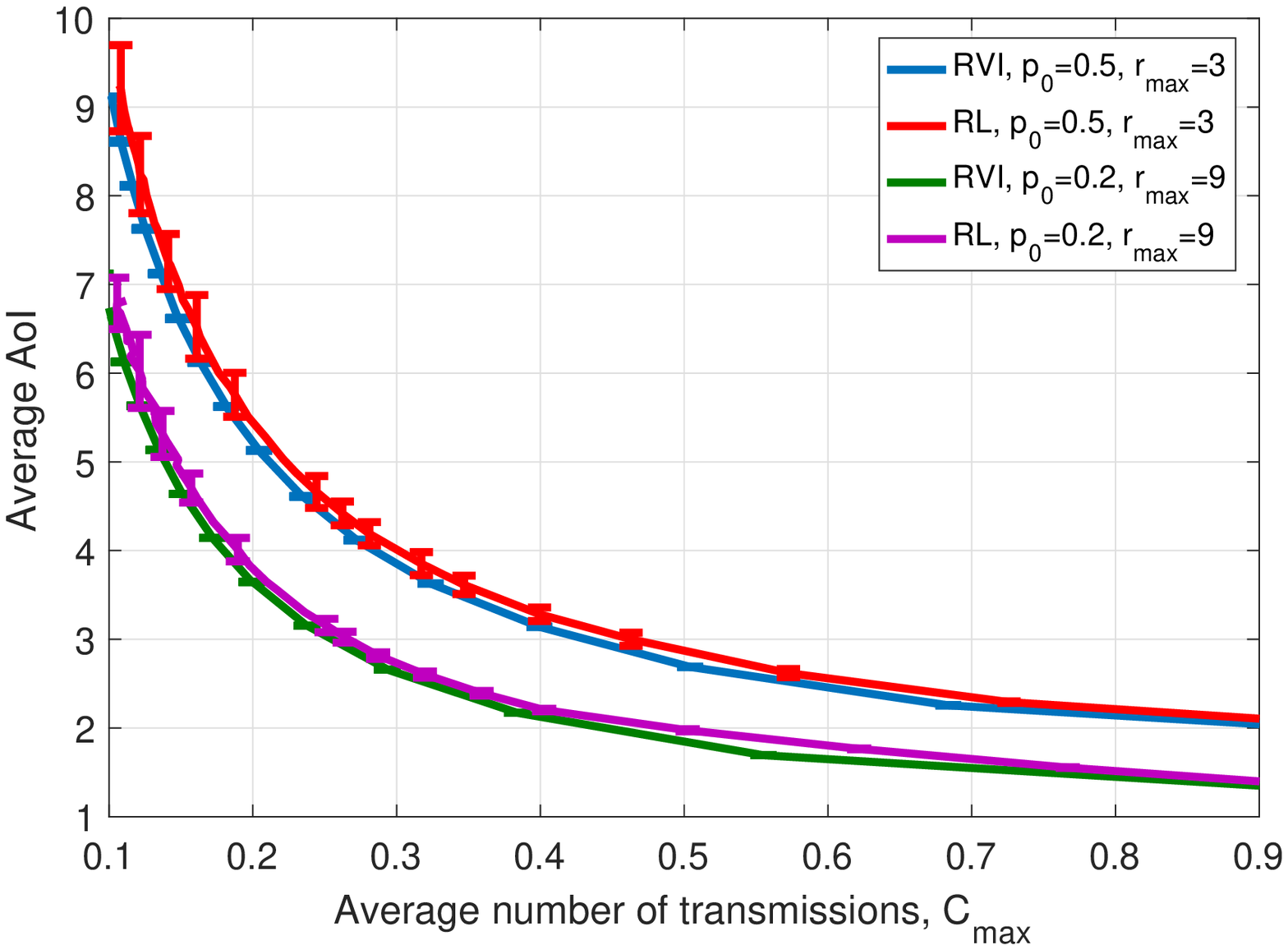}
\caption{Performance of the proposed RL algorithm (average-cost SARSA) and its comparison with the RVI algorithm for $n=10000$ iterations, and values are averaged over 1000 runs for different $p_0$ and $r_{max}$ values when $\lambda=0.5$ (both the mean and the variance are shown).}
\label{fig:learn2}
\end{figure}

\section{Conclusions}\label{sec:conclusion}

We have considered a communication system transmitting time-sensitive data over an imperfect channel with the average AoI as the performance measure, which quantifies the timeliness of the data available at the receiver. Considering both the classical ARQ and the HARQ protocols, preemptive scheduling policies have been proposed by taking into account retransmissions under a resource constraint.  In addition to identifying a randomized threshold structure for the optimal policy when the error probabilities are known, an efficient RL algorithm is also presented for practical applications when the system characteristics may not be known in advance. The effects of feedback and the HARQ structure on the average AoI are demonstrated through numerical simulations. The algorithms adopted in this paper are also relevant to different  systems concerning the timeliness of information, and the proposed methodology can be used in other CMDP problems. As future work, the problem will be extended to time-correlated channel statistics in a multi-user setting.

\appendix




\subsection{Verifying the assumptions of \cite{sennott_1993}} 
\label{Append_assumptions}

In this section, we show that the assumptions for the main results of \cite{sennott_1993} are satisfied for Problem~\ref{problem}. We start with a few standard definitions about Markov chains:
In a Markov chain with a countable state space $\mathcal{S}$, a state $s \in \mathcal{S}$ is called \emph{positive recurrent} if the expected number of transitions needed to return to state $s$ given that the chain started in state $s$ is finite. A communication class $Z \subset \mathcal{S}$ is defined as a subset of the state space $\mathcal{S}$ such that all states within it communicate; that is, for any $s,s' \in Z$, starting from state $s$ the chain reaches state $s'$ with some positive probability.  A communication class is positive recurrent if and only if all states in a communication class are positive recurrent. 
\cite{Ross2006}

We continue with Definition~2.3 of \cite{sennott_1993}: Let $G \subset \mathcal{S}$ be a nonempty set of states of a CMDP. Given a state $s \in \mathcal{S}$, let $\mathcal{R}(s, G)$ be the class
of policies such that $P^{\pi}(s_t \in G \textrm{ for some } t\geq 1 ~| s_0 = s) = 1$ and the expected
time $m_{s,G}(\pi)$ of the first passage from $s$ to $G$ under $\pi$ is finite. Let $\mathcal{R}^*(s, G)$ be the class of policies $\pi \in \mathcal{R}(s, G)$ such that, in addition, the expected average AoI  $c_{s,G}(\pi)$ and the expected transmission cost $d_{s,G}(\pi)$ of a first passage from $s$ to $G$ are finite.

\begin{prop}
The following hold for Problem~\ref{problem}:
\begin{enumerate}[(i)]
\item For all $b>0$, the set $G(b)\triangleq \{s|\textrm{ there exists an action a such that } c(s,a)+d(s,a)\leq b\}$ is finite (Assumption~1 of \cite{sennott_1993}).
\label{it:i}
\item There exists a deterministic policy $\pi$ that induces a Markov chain with the following properties: the state space $\mathcal{S}^\pi$ consists of a single (nonempty)
positive recurrent class $R$ and a set $U$ of transient states such that $\pi \in R^*(s, R)$, for any $s \in U$, and both the average AoI $J^{\pi}$ and the average transmission cost $C^{\pi}$ on $R$ are finite (Assumption~2 of \cite{sennott_1993}).
\label{it:ii}
\item Given any two states $s, s' \in \mathcal{S}$, there exists a policy $\pi$ (a function of $s$ and $s'$) such that $\pi \in \mathcal{R}^*(s, \{s'\})$ (Assumption~3 of \cite{sennott_1993}).
\label{it:iii}
\item If a deterministic policy has at least one positive recurrent state, then it has a single positive recurrent class, and this class contains the state $(1,0)$ (Assumption~4 of \cite{sennott_1993}).
\label{it:iv}
\item There exists a policy $\pi$ such that $J^{\pi} < \infty$ and $C^{\pi} < C_{max}$ (Assumption~5 of \cite{sennott_1993}).
\label{it:v}
\end{enumerate}
\label{lemma_assump}
\end{prop}
\begin{proof}
Note that \eqref{it:i}-\eqref{it:iv} are independent from the constraint \eqref{eq:constraint}, and the policies required in the proposition need not be deterministic unless specifically required.

First note that \eqref{it:i} holds trivially, since for any $b$, if state $(\delta,r) \in G(b)$ then $r < \delta \le b$ by \eqref{eq:S}.
 
To prove \eqref{it:ii}, consider the policy $\pi(\delta,r)=\new$ for all $(\delta,r) \in \mathcal{S}$. Since $0<g(0)<1$, $R=\{(1,0)\} \cup \{(\delta,1): \delta=1,2,\ldots,\}$ is a recurrent class since from any state $(\delta,r) \in R$, the next state is either $(1,0)$ or $(\delta+1,1)$, both belonging to $R$. Furthermore, the set of states $U=\mathcal{S}\setminus R$ is clearly transient: starting from any $s \in U$, the probability of not getting to state $(1,0)$ (and hence to $R$) in at most $k$ steps is $g(0)^k$. The latter also implies that $\pi \in R^*(s,R)$. Finally, $C^\pi=1$, and $J^\pi = 1/(1-g(0))$, proving \eqref{it:ii}.

To prove \eqref{it:iii}, let $s=(\delta,r)$ and $s'=(\delta',r')$. For any $s,s'$, we construct the required policy. 
First note that from $(1,0)$, we can govern the state with positive probability to any valid state $(\delta',r')$ by being idle in states
$(\delta'',0)$ for $\delta'' < \delta'-r'$, sending a new packet in state $(\delta'-r',0)$, and retransmitting in states $(\delta'-r'+k,k)$ for $k=1,\ldots,r'-1$. Sending a new packet in any other state $(\delta'',r'')$ will send the chain to state $(1,0)$ as quickly as possible, with the number of steps being exponentially distributed with parameter $g(0)$. It is trivial to see that the proposed policy belongs to $\mathcal{R}^*(s, \{s'\})$.

To see \eqref{it:iv}, notice that the only way the AoI does not increase in one step is if there is a successful transmission, after which the chain returns to state $(1,0)$. Thus, any (positive) recurrent class must contain the state $(1,0)$; and hence, there can only be a single positive recurrent class.

Finally, it is easy to see that the policy $\pi$ defined as $\pi(\delta,r)=\new$ if $\delta-1$ is a multiple of  $2 \lceil 1/C_{max}\rceil$, and $\pi(\delta,r)=\idle$ otherwise, satisfies the requirements of \eqref{it:v}: $C^\pi = 1/(2 \lceil 1/C_{max}\rceil) \le C_{max}/2<C_{max}$, and $J^\pi < \infty$ since $P^\pi(\delta> 2k \lceil 1/C_{max}\rceil) 
= g(0)^k$ for any $k \ge 0$. This completes the proof of the proposition.
\end{proof}

\subsection{Proof of Lemma~\ref{lem:ARQ}} 
\label{AppendA}

We are going to show that the decision to transmit ($a=\new$) is monotone with respect to the age $\delta$, that is if $a^*(\delta^1)=\new$, then $a^*(\delta^2)=\new$ for all $\delta^2 \ge \delta^1$. By \eqref{eq:opt_eta}, this holds if $Q_{\eta}(\delta,a)$ has a \textit{sub-modular} structure \cite{Topkis1978}: that is, when the difference between the $Q$ functions is monotone with respect to the state-action pair ($\delta,a$). We have
\begin{align}
Q_{\eta}(\delta^1,\new)-Q_{\eta}(\delta^1,\idle)\geq Q_{\eta}(\delta^2,\new)-Q_{\eta}(\delta^2,\idle), \label{eq:submodular}
\end{align} 
for any $\delta^2\geq \delta^1$. From (\ref{eq:update}) and (\ref{eq:update1}), for any $\delta>0$, we have
\begin{align}
Q_{\eta}(\delta,\new)-Q_{\eta}(\delta,\idle)&=\eta+(1-p) h_{\eta}(1)-(1-p) h_{\eta}(\delta+1).
\end{align}
We can see that (\ref{eq:submodular}) holds if and only if $h_{\eta}(\delta)$ is a non-decreasing function of the age. 

We compare the costs incurred by the systems starting in states $\delta^1$ and $\delta^2$ via coupling the stochastic processes governing the behavior of the system; that is, we assume that the realization of the channel behavior is the same for both systems over the time horizon (this is valid since channel states/errors are independent of the ages and the actions). Assume a sequence of actions $\{a^2_t\}_{t=1}^{\infty}$ corresponds to the optimal policy starting from age $\delta^2$ for a particular realization of channel errors, and let $\{\delta^i_t\}$ denote the sequence of states obtained after following actions $\{a^2_t\}$ starting from state $\delta_1=\delta^i$, $i=1,2$. Then, if $\delta^1 \le \delta^2$, clearly $\delta^1_t \le \delta_2^t$ for all $t$.
Furthermore, by the Bellman optimality equation \eqref{eq:Bellman},
\begin{align*}
h_\eta(\delta^1) & \le \delta^1_1+\eta \cdot \mathbbm{1}[a^2_1 \neq \idle]-L^*_\eta + \Exp{h_\eta(\delta^1_2) } \\
& \le \delta^1_1+ \eta  \cdot \mathbbm{1}[a^2_1 \neq \idle]-L^*_\eta  + \Exp{ \delta^1_2 + \eta \cdot \mathbbm{1}[a^2_2 \neq \idle]-L^*_\eta + \Exp{h_\eta(\delta^1_3) }} \\
& \quad \vdots \\
& \le \Exp{\sum_{t=1}^\infty (\delta^1_t+ \eta \cdot \mathbbm{1}[a^2_t \neq \idle]-L^*_\eta) \bigg| \delta^1_1=\delta^1} \\
& \le \Exp{\sum_{t=1}^\infty (\delta^2_t+ \eta \cdot \mathbbm{1}[a^2_t \neq \idle]-L^*_\eta) \bigg| \delta^1_1=\delta^2} \\
& = h_\eta(\delta^2)~.
\end{align*}
This completes the proof of the lemma.
\qed




\subsection{Proof of Lemma~\ref{thm1}}
\label{AppendB}
First we compute the steady state probabilities $p_\delta$ of the age $\delta$ for a given integer threshold $\Delta$, for all $\delta= 1,2,\ldots, N$. We have
\begin{align*}
p_{\delta}&= \begin{cases} p_1 & \text{ if }  1 \le \delta \le \Delta \\
p_{\delta-1}p=p_1p^{\delta-\Delta} & \text{ if } \delta \ge \Delta+1~.
\end{cases}
\end{align*} 
Since $\sum_{\delta=1}^{\infty}{p_{\delta}}=1$, we can compute the $p_{\delta}$ in closed form when $N$ goes to infinity:
\begin{align}
\label{eq:s_prob}
p_{\delta}=\begin{cases}
\frac{1}{\Delta+\frac{p}{1-p}}~&\mbox{if } \delta\leq \Delta;\\
\frac{p^{\delta-\Delta}}{\Delta+\frac{p}{1-p}}~ &\mbox{otherwise}.
\end{cases}
\end{align}

Then, the closed form of the expected Lagrangian cost function can be computed as:
\begin{align}
\label{eq:JetaARQ}
L^{\Delta}_{\eta}&=\sum_{\delta=1}^{\infty}p_{\delta}(\delta + \eta \mathbbm{1}[\delta \ge \Delta])= p_1 \left(\sum_{\delta=1}^{\Delta-1}{\delta}+ \sum_{\delta=\Delta}^{\infty}{p^{\delta-\Delta}(\delta+\eta)}\right) \nonumber \\
&=p_1 \left(\frac{(\Delta-1)\Delta}{2}+ \frac{\eta+\Delta}{1-p}+\frac{p}{(1-p)^2}\right). 
\end{align}
Substituting $p_1$ from \eqref{eq:s_prob} and minimizing over $\Delta$ (by setting the derivative $\partial L^{\Delta}_{\eta}/\partial \Delta$ to zero) yields that the optimal non-integer value of $\Delta$ is given by
\begin{align*}
\hat{\Delta}_\eta&= \frac{\sqrt{2\eta(1-p)+p}-p}{1-p}~.
\end{align*}  
Using that $L^{\Delta}_\eta$ is a convex function of $\Delta$ by \eqref{eq:JetaARQ}, the optimal integer threshold $\Delta^*_\eta$ is either
\[
\left\lfloor\frac{\sqrt{2\eta(1-p)+p}-p}{1-p} \right\rfloor \text{ or } \left\lceil\frac{\sqrt{2\eta(1-p)+p}-p}{1-p} \right\rceil~.
\]
%
%
Computing just the cost term from \eqref{eq:JetaARQ}, we obtain the formula for $C^{\Delta}$ for any integer threshold $\Delta$.
\qed

\subsection{Proof of Theorem~\ref{thm:ARQ}}
\label{AppendC}
Let $\pi_{\eta^*}$ denote the deterministic solution of the Bellman equation \eqref{eq:Bellman}. If $C^{\pi_{\eta^*}}=C_{max}$ then $\pi_{\eta^*}$ is the optimal solution to Problem~\ref{problem} (by Proposition~3.2 and Lemma~3.10 of \cite{sennott_1993}), and since it is a threshold policy by Lemma~\ref{lem:ARQ} with threshold $\Delta_{C_{max}}=\Delta_1=\Delta_2$ (as can be obtained by inverting  equation~\eqref{eq:const}), the theorem holds.

For $C^{\pi_{\eta^*}} \neq C_{max}$, we first show that the optimal policy is a mixture of two threshold policies that differ in at most a single state, based on the construction used to prove Theorem~2.5 of \cite{sennott_1993}.
Assume without loss of generality that $C^{\pi_{\eta^*}} > C_{max}$, and consider a sequence of Lagrange multipliers $\eta_n \downarrow \eta^*$ such that the  corresponding deterministic solutions $\pi^*_{\eta_n}$ of \eqref{eq:Bellman} (which are also $\eta_n$-optimal by Proposition~3.2 of \cite{sennott_1993}) converge to a policy $\pi^*$.\footnote{If $C^{\pi_{\eta^*}}>C_{max}$, $\eta_n$ should be increasing to $\eta^*$, and the rest of the proof follows the same lines as for the case of $C^{\pi_{\eta^*}} < C_{max}$.}
By Lemma~\ref{thm1}, these are all threshold policies, and so $\pi_{\eta^*}$ and $\pi^*$ are both threshold policies. By Lemma~3.7 (iii) of \cite{sennott_1993}, $\pi^*$ is $\eta^*$-optimal, and $C^{\pi^*} \le C_{max}$ by Lemma~3.4 of \cite{sennott_1993}. If $C^{\pi^*}=C_{max}$ then the proof can be completed as in the case of $C^{\pi_{\eta^*}}=C_{max}$. Thus, we are left with the case of $C^{\pi^*} < C_{max}$. 
Denoting by $(\mu,\pi_{\eta^*},\pi^*)$ the randomized policy that selects $\pi_{\eta^*}$ with probability $\mu=\frac{C_{max}-C^{\pi^*}}{C^{\pi_{\eta^*}}-C^{\pi^*}}$ and $\pi^*$ with probability $1-\mu$ before the system starts and then uses the selected policy forever, it follows that $(\mu,\pi_{\eta^*},\pi^*)$ has average transmission cost $C_{max}$, while it is also $\eta^*$-optimal by Lemma~3.9 of \cite{sennott_1993}. Therefore,
by Lemma~3.10 of \cite{sennott_1993}, $(\mu,\pi_{\eta^*},\pi^*)$ is an optimal solution to Problem~\ref{problem}. 

Next we show that the thresholds of the two policies are $\Delta_1$ and $\Delta_2$. 
From  the proof of Lemma~\ref{thm1}, one can easily deduce that the average AoI of a threshold policy for any integer threshold $\Delta$ is given by
\[
J^\Delta= \frac{(\Delta(1-p)+p)^2+p}{2(1-p)(\Delta(1-p)+p)}+\frac{1}{2}.
\]
Expressing $J^\Delta$ as a function of $C^\Delta$ (given in \eqref{eq:const}), and extending the definition of $C^\Delta$ and $J^\Delta$ to positive real values of $\Delta$,  one can see that 
\[
J^\Delta = \frac{1}{2(1-p)C^\Delta}+\frac{1}{2}+\frac{p C^\Delta}{2(1-p)}
\]
is a convex function of $C^\Delta$. Denoting the threshold of $\pi_{\eta^*}$ and $\pi^*$ by $\Delta_{\eta^*}$ and $\Delta^*$, respectively, we obviously have that the expected average AoI of $(\mu,\pi_{\eta^*},\pi^*)$ is
$\mu J^{\Delta_{\eta^*}} + (1-\mu) J^{\pi^*}$, while the expected average transmission cost is $C_{max}$. By the convexity of $J^\Delta$, and since $C^{\Delta^*}= C^{\pi^*}  < C_{max} < C^{\pi_{\eta^*}} = C^{ \Delta_{\eta^*}}$ it follows that the integer threshold values minimizing the AoI
must be the closest integers (from above and below) to $\Delta_{C_{max}}$, the minimizer of $J^{\Delta}$ over the reals.
That is, $\Delta_{\eta^*} =  \Delta_1 = \lfloor \Delta_{C_{max} } \rfloor$ and $\Delta^*=\Delta_2 = \lceil \Delta_{C_{max} } \rceil$ (recall that the transmission cost $C^\Delta$ is a decreasing function of the threshold $\Delta$). Note that this also implies that $\mu=\mu^*$ (recall that $\mu^*$ is specified in the statement of the theorem).

To complete the proof, define $(\pi_{\eta^*},\pi^*)-\mu$ to be the policy that randomly selects between  $\pi_{\eta^*}$ and $\pi^*$ every time state $(1,0)$ is reached (independently, and with probability $\mu$ and $1-\mu$, resp.) and follows the chosen policy until state $(1,0)$ is reached again. Since $(1,0)$ is a positive recurrent state of both $\pi^*$ and $\pi_{\eta^*}$, the policy $(\pi_{\eta^*},\pi^*)-\mu$ has the same expected AoI and transmission cost as $(\mu,\pi_{\eta^*},\pi^*)$, which randomizes once at the beginning. Therefore, $(\pi_{\eta^*},\pi^*)-\mu$ is optimal. Moreover, since $\pi^*$ and $\pi_{\eta^*}$ only differ in state $\delta=\Delta_2$, the randomization can be performed only in that state. Thus, since $\mu=\mu^*$, the policy $(\pi_{\eta^*},\pi^*)-\mu$ is identical to $\pi^*_{C_{max},\mu^*}$, defined in the theorem, proving that $\pi^*_{C_{max},\mu^*}$ is optimal.
\qed

\subsection{Unconstrained case ($C_{max}=1$)}
\label{sec:unconstrained}

Here we analyze the problem without a transmission constraint, that is, when $C_{max}=1$. We show that the conditions of part (ii) of the Theorem in \cite{sennott_1989} hold, implying that there exists a deterministic optimal policy satisfying the Bellman equation \eqref{eq:Bellman} with $\eta=0$ and $a$ restricted to $\{\new,\retx\}$, namely
\begin{align}
\label{eq:Bellman2}
h(\delta,r)+L^*&=\min_{a\in\{\new,\retx\}}\big(\delta+\Exp{h(\delta',r')}\big)
\end{align}
for some function $h(\delta,r)$ and constant $L^*$.

For any $\alpha \in (0,1)$, policy $\pi$ (here a policy is an arbitrary, possibly randomized decision strategy that may depend on the whole history) and initial state $s_0$, let
\[
J^\pi_\alpha(s_0) \triangleq  \Exp{\sum_{t=0}^\infty \alpha^t \delta^{\pi}_t \Big|s_0}~,
\]
and $J_\alpha(s_0)= \inf_{\pi} J^\pi_\alpha(s_0)$.

Consider policy $\pi_{\new}$, which transmits a new update in every step. One can verify (e.g., by induction) that the stationary distribution of the Markov chain induced by this policy is a geometric distribution over states $(\delta,0)$ with parameter $1-p$, where $p=g(0)$: that is, the probability of being in state $(\delta,0)$ is $(1-p)p^{\delta-1}$.

Next, we verify Assumption~1 of \cite{sennott_1989}, which requires that
for any $\alpha$ and state $s=(\delta_0,r)$, $J_\alpha(s)$ is finite.
Note that given the first state is $(\delta_0,r)$, we have $\delta_t \le \delta_0 + t$. Therefore,
\[
J_\alpha((\delta_0,r)) \le J_\alpha^{\pi_{\new}}((\delta_0,r)) \le \sum_{t=0}^\infty \alpha^t (\delta_0+t) < \infty,
\]
which is what we wanted to prove.

Let $h_\alpha(s)=J_\alpha(s) - J_\alpha(s_0)$, where $s_0=(1,0)$. In what follows, we give upper and lower bounds on $h_\alpha$.
Consider an arbitrary policy $\pi$ starting from state $s=(\delta, r)$. Since in every time step a transmission is successful with probability at least $1-p$ (since the success probability cannot decrease with retransmissions), and if two successive transmissions are successful, the second must be a new update, in two steps the process returns to state $s_0$ with probability at least $(1-p)^2$. Thus, if the MDP is started from $s$ and $s_0$, with probability at least $q=(1-p)^4$, they synchronize after two steps, after which the terms in the summations defining $J^\pi_\alpha(s)$ and $J^\pi_\alpha(s_0)$ become identical: denoting the AoI at time $t$ by $\delta_t(s)$ and $\delta_t(s_0)$ for the processes started in state $s$ and $s_0$, respectively, and by $T$ the first time step they simultaneously reach the same state,
we have $\delta_t(s) - \delta_t(s_0) \le \delta+t-1$ for $t < T$ (before the synchronization happens) and $\delta_t(s)=\delta_t(s_0)$ for $t \ge T$ (after synchronization). By our argument above, for any $k$, $\Pr(T \ge 2k) \le (1-q)^k$, and so
\begin{align*}
\lefteqn{\Exp{J^\pi_\alpha(s) - J^\pi_\alpha(s_0)} = \Exp{\sum_{t=0}^\infty \alpha^t (\delta_t(s) - \delta_t(s_0))} }\\
& = \sum_{t=0}^\infty \mathbb{E}\Big[\alpha^{2t} (\delta_{2t}(s) - \delta_{2t}(s_0)) +  \alpha^{2t+1} (\delta_{2t+1}(s) - \delta_{2t+1}(s_0)) \Big| 2t+1 < T \Big] \Pr(2t+1<T) \\
& \le \sum_{t=0}^\infty \left(\alpha^{2t}(\delta+2t-1) + \alpha^{2t+1} (\delta+2t)\right) (1-q)^{t+1} \\
& < 2 \delta \sum_{t=0}^\infty (1-q)^{t+1} + 4 \sum_{t=0}^\infty t(1-q)^{t+1} \\
& = \frac{2(1-q)}{q} \delta+ \frac{4(1-q)^2}{q^2}
\end{align*}
Therefore, we have 
\[
h_\alpha(s) \le \sup_\pi \Exp{J^\pi_\alpha(s) - J^\pi_\alpha(s_0)} < \frac{2(1-q)}{q} \delta+ \frac{4(1-q)^2}{q^2} \triangleq M_\delta.
\]
Similarly, since $\delta_t(s_0)-\delta_t(s) \le t$, we can prove that $\Exp{J_\alpha(s_0) - J_\alpha(s)}  < \frac{4(1-q)}{q^2}$, which implies
$-\frac{4(1-q)}{q^2} \le h_\alpha(s)$. 
The latter directly proves Assumption~2 of \cite{sennott_1989}, which requires $h_\alpha(s)$ to be uniformly bounded from below by a nonpositive constant for all $\alpha \in (0,1)$ and state $s$. 

Finally, for any starting state $s=(\delta,r)$ let $s'=(\delta',r')$ denote the next state following action $a$. Assumptions~3 and~3$^*$ of \cite{sennott_1989} are satisfied if
$\Exp{M_{\delta'}|s, a}<\infty$ holds for all $s$, $s'$ and $a$. Since for any $s$ and $a$ there can be only two states $s'$ with non-zero probability and all $M_{\delta'}$ are finite, this is trivially satisfied.

Therefore, Assumptions~1--3$^*$ of Sennott \cite{sennott_1989} are satisfied, and hence part (ii) of her Theorem implies that there exists a deterministic, optimal policy satisfying the Bellman equation \eqref{eq:Bellman2} (equivalently, equation~\eqref{eq:Bellman} with $\eta=0$ and $a$ restricted to $\{\new,\retx\}$).

\section*{Acknowledgement}
The authors would like to thank the anonymous reviewers for their careful reading of the paper and their insightful comments.

\bibliography{ageof8}

\end{document}